\documentclass[aps,twocolumn,floatfix]{revtex4-2}

\usepackage{amsmath,amsthm,amsfonts,amssymb,times,bbm,graphicx}
\usepackage{physics}
\usepackage{hyperref}

\newtheorem{theorem}{Theorem}

\newtheorem{lemma}[theorem]{Lemma}
\newtheorem{corollary}[theorem]{Corollary}
\newtheorem{definition}[theorem]{Definition}

\def\CC{\mathbbm{C}}

\def\NN{\mathbbm{N}}

\def\H{\mathcal{H}}

\def\Id{\mathbbm{1}}

\DeclareMathOperator{\range}{range}

\def\A{\mathcal{A}}
\def\B{\mathcal{B}}
\def\C{\mathcal{C}}
\def\D{\mathcal{D}}
\def\F{\mathcal{F}}
\def\G{\mathcal{G}}

\def\K{\mathcal{K}}

\begin{document}

\title{The inflation hierarchy and the polarization hierarchy are complete for the quantum bilocal scenario}

\author{Laurens T.\ Ligthart}
\author{David Gross} 
\affiliation{ Institute for Theoretical Physics, University of Cologne, Germany } 
\date{\today} 

\begin{abstract} 
	It is a fundamental but difficult problem to characterize the set of correlations that can be obtained by performing measurements on quantum mechanical systems.
	The problem is particularly challenging when the preparation procedure for the quantum states is assumed to comply with a given \emph{causal structure}.
	Recently, a first completeness result for this \emph{quantum causal compatibility problem} has been given, based on the so-called \emph{quantum inflation technique}.
	However, completeness was achieved by imposing additional technical constraints, such as an upper bound on the Schmidt rank of the observables.
	Here, we show that these complications are unnecessary in the \emph{quantum bilocal scenario}, a much-studied abstract model of entanglement swapping experiments.
	We prove that the quantum inflation hierarchy is complete for the bilocal scenario in the commuting observables model of locality.
	We also give a bilocal version of an observation by Tsirelson, namely that in finite dimensions, the commuting observables model and the tensor product model of locality coincide.
	These results answer questions recently posed by Renou and Xu.
	Finally, we point out that our techniques can be interpreted more generally as giving rise to an SDP hierarchy that is complete for the problem of optimizing polynomial functions in the states of operator algebras defined by generators and relations.
	The completeness of this \emph{polarization hierarchy} follows from a quantum de~Finetti theorem for states on maximal $C^*$-tensor products.
\end{abstract} 

\maketitle

\section{Introduction} \label{sec:intro}

Studying the correlations that can be obtained by performing measurements on spatially separated systems is central to the theory of quantum information.
Indeed, such considerations led to the development of Bell inequalities and the theory of non-locality \cite{bell1964einstein, bell2004speakable, brunner2014bell}. 
These phenomena require a theory of Nature that is fundamentally different from a classical theory.
That is, if correlations are observed that break a Bell inequality or are otherwise shown to be non-local, they cannot be derived from a model with local hidden variables.
Such correlations have indeed been observed in numerous experiments (see e.g.~Refs.~\cite{brunner2014bell,pan2012multiphoton} for reviews).

Practically, such non-local correlations have many applications in quantum information processing tasks, such as
quantum cryptography \cite{acin2007device,barrett2005no},
private random number generation \cite{pironio2010random},
entanglement detection \cite{bancal2011device}
and quantum networks \cite{branciard2012bilocal, rosset2016nonlinear, pozas2022full, tavakoli2022bell}.

This paper will focus on the \emph{bilocal scenario}.
In the bilocal scenario (Fig.~\ref{fig:bilocal_scenario}), we are concerned with the set of correlations that can be obtained by three parties (Alice, Bob, and Charlie) performing measurements on pairs of quantum particles originating from two independent sources: 
One distributing a pair between Alice and Bob, and one between Bob and Charlie.
We assume that each party can choose among a finite number of measurement settings, their choices being labeled by numbers $x,y,z$.
Each then obtains one of a finite number of possible outcomes.
We represent their respective outcomes by $\alpha,\beta,\gamma$.
The statistics of such an experiment are then described by a collection $p(\alpha\beta\gamma|xyz)$ of conditional probabilities.

The bilocal scenario is one of the most fundamental \emph{causal structures}:
It is the simplest non-trivial structure in which source states are assumed to be independent.
It is also a straightforward generalization of the Bell scenario. 
Nevertheless it allows for new behaviour such as entanglement swapping \cite{zukowski1993swapping} and is surprisingly hard to analyze.
Here, we are primarily concerned with the \emph{bilocal causal compatibility problem}:
Given a collection of conditional probabilities $p(\alpha\beta\gamma|xyz)$, decide whether it is compatible with an experiment of the form described above.

Several techniques to answer this question have already been developed.
These include, but are not limited to, (non-linear) Bell inequalities \cite{branciard2012bilocal}, 
machine learning techniques \cite{canabarro2019machine},
information-theoretic methods \cite{chaves2015information},
scalar extension \cite{pozas2019quantum, pozas2019bounding, renou2022two} and the inflation technique that is also considered in this paper \cite{wolfe2019inflation, navascues2020inflation, wolfe2021quantum, ligthart2021convergent}.
For a more complete list, both on the bilocal scenario and more general network scenarios, we refer to the excellent review of Ref.~\cite{tavakoli2022bell}.

Recently, the authors of Ref.~\cite{renou2022two} asked whether the hierarchy of semi-definite programming constraints known as the \emph{quantum inflation technique} is complete for the bilocal compatibility problem.
One of the main results of this paper, partly building on their constructions, is to answer this question in the affirmative.
We develop two complete semidefinite programming hierarchies that are closely related.
The first, which we call the \emph{polarization hierarchy}, uses symmetric product states to linearize the non-convex independence constraint.
The second is a version of the quantum inflation hierarchy, which relaxes the independence condition to a family of linear symmetry constraints.
Along the way, we obtain a number of equivalent characterizations of bilocal quantum correlations, which might be of independent interest.

\subsection{Outline}

The paper is structured as follows.
In Section \ref{sec:prelims} several technical preliminaries are explained that are useful to understand later discussions on models of locality and the algebraic formulation of quantum theory.

Section \ref{sec:Quantum correlations} discusses different notions of quantum correlations, first for the Bell scenario and then for the bilocal scenario.
Theorem \ref{thm:fact to biloc} of this section is central to proving equivalence of several of these formulations for the bilocal scenario.
Ultimately, it explains why complete hierarchies for the bilocal scenario are easier to construct than in the general case treated in Ref.~\cite{ligthart2021convergent}.

Section \ref{sec:complete hierarchy} constructs  two complete semidefinite programming hierarchies for the reduced model of the bilocal scenario.
This part of the paper can be understood as an exposition and slight adaption of the methods developed in Ref.~\cite{ligthart2021convergent} (which in turn builds on Refs.~\cite{navascues2008convergent, pironio2010convergent, wolfe2021quantum, navascues2020inflation, raggio1989quantum}).

\section{Preliminaries} \label{sec:prelims}

\subsection{Quantum models of locality}

In order to give a precise definition of the set of \emph{bilocal quantum correlations}, one needs to fix a quantum model of locality.
This turns out to be a surprisingly subtle issue.

There are two commonly used ``pictures'' on which a formalization of quantum mechanical descriptions of Nature can be based.

In elementary quantum mechanics (related to the Schr\"odinger picture), the fundamental mathematical object associated with a quantum system is a Hilbert space $\H$.
The set of observables is then derived as the algebra of bounded operators $B(\H)$ acting on $\H$.

Alternatively, in \emph{algebraic quantum mechanics} \cite{bratteli2012operator, landsman2017foundations, moretti2019fundamental} (related to  the Heisenberg picture), quantum systems are primarily described via an algebra $\A$ of observables.
A Hilbert space is then a secondary object, which can be derived e.g.\ via the GNS construction \cite{blackadar2006operator}.

The two points of view are mostly equivalent as a basis for describing natural phenomena.
Differences are commonly associated with finer technical points, e.g.\ in the rigorous description of the thermodynamic limit \cite{bratteli2012operator}.
One would thus assume that the choice of which point of view to adopt becomes a matter of taste and convenience.
While most working physicists prefer the Schr\"odinger picture, 
the algebraic model is easier to reason about algorithmically, which explains its use in completeness proofs such as those of Refs.~\cite{navascues2008convergent, pironio2010convergent, ligthart2021convergent}.

However, the two approaches suggest different formalizations of the notion of ``locality'', which is obviously relevant for the problem treated in this paper.

Indeed, consider two spatially separated subsystems $A$, $B$ of some composite system.
Separation implies that physical properties of $A$ and $B$ can be simultaneously measured, which means that the associated observable algebras $\A$, $\B$ must mutually commute, $[a,b]=0, a\in\A, b\in \B$.
In algebraic quantum mechanics, this 
assumption (sometimes referred to as \emph{Einstein locality} \cite[Sec. 8.5]{landsman2017foundations}) 
is the only one made.

In contrast, the Schr\"odinger picture-approach is to associate one Hilbert space $\H_A, \H_B$ with each subsystem and to take the observable algebras to be 
\begin{align}\label{eqn:tp algebras}
	\begin{split}
		\A &= B(\H_A)\otimes \Id \subset B(\H_A\otimes\H_B), \\
		\B &= \Id\otimes B(\H_B) \subset B(\H_A\otimes\H_B)
	\end{split}
\end{align}
respectively.

The surprisingly technically complex theory of tensor products of operator algebras \cite{takesaki1} shows that not every pair of commuting algebras can be realized on a tensor product of Hilbert spaces as in Eq.~(\ref{eqn:tp algebras}).
For a considerable time, it was an open question (known as \emph{Tsirelson's Problem} \cite{scholz2008tsirelson,junge2011connes,fritz2012tsirelson}), whether these operator-theoretic subtleties would manifest themselves at the level of finite sets of observable correlation functions (as made precise in Sec.~\ref{sec:bell}).
Unfortunately, it has now become clear that this is indeed the case \cite{ji2020mip}.
Thus, whenever one speaks about ``quantum correlations'', one has to be specific as to whether one is working in the more restrictive tensor product Hilbert space model or the more general commuting observable model.

At present, there does not seem to be strong evidence indicating which of the two approaches is more relevant for the description of natural phenomena.
Both are legitimate targets of inquiry, as long as authors indicate clearly (as we have tried to do) which model they are working with at any time.

\subsection{Algebras of observables}

For a Hilbert space $\H$, let $B(\H)$ be the set of bounded operators on $\H$.
An algebra $\A\subset B(\H)$ of operators that is closed under taking adjoints and under operator norm limits is a \emph{concrete $C^*$-algebra} \cite{blackadar2006operator, kadison1983fundamentals}.

The same way one can axiomatically define the notion of a \emph{group} as an abstraction of concrete groups of linear operators, one can also define $C^*$-algebras abstractly, without referring to a concrete Hilbert space.
We will encounter abstract $C^*$-algebras in the context of the NPO hierarchy in Sec.~\ref{sec:complete hierarchy}, but will keep the discussion of this theory at a minimum.
See Refs.~\cite{blackadar2006operator, kadison1983fundamentals, takesaki1, bratteli2012operator} for details.

A $C^*$-algebra is \emph{unital} if it contains an element acting as the identity.
All $C^*$-algebras encountered in this paper are unital, and we will not mention this property explicitly.
What is more, whenever $\A$ is a $C^*$-subalgebra of some $C^*$-algebra $\D$, we will assume that $\A$ contains the unit of $\D$.

An element $a$ of a $C^*$-algebra $\A$ is \emph{positive} if it is of the form $a=b^*b$ for some $b\in \A$.
A \emph{state} on $\A$ is a linear functional that is \emph{positive} in the sense that $\rho(b^* b) \geq 0$ for every $b\in\A$ and \emph{normalized} in the sense that $\rho(\Id) = 1$.
We denote the \emph{state space} of an algebra $\A$ by $K(\A)$.
A \emph{positive operator-valued measure} (POVM) with finitely many outcomes labeled by a variable $\alpha$ is a set $\{A_\alpha\}_\alpha\subset\A$ such that $A_\alpha$ is positive and $\sum_\alpha A_\alpha = \Id$.

Occasionally, it is necessary to use a weaker notion of convergence than the one induced by the operator norm.
Recall from elementary real analysis, that, in addition to norm convergence, the weaker concept of point-wise convergence has its role.
A non-commutative analogue of the topology of point-wise convergence is the \emph{weak operator topology}.
Here, instead of evaluating functions at points in their domain, one takes ``matrix elements'' of operators between state vectors.
More precisely, a net $a_\lambda\subset B(\H)$ converges in the weak operator topology if $\langle\phi|a_\lambda|\psi\rangle$ converges in $\CC$ for all vectors $|\phi\rangle,|\psi\rangle\in\H$.
Concrete $C^*$-algebras that are closed in the weak operator topology (rather than just the norm topology) are \emph{von~Neumann algebras}.

Weak operator closures appear naturally in the study of tensor products of operators acting on the tensor product of infinite-dimensional Hilbert spaces.
For example, consider two Hilbert spaces $\H_A, \H_B$, and define the \emph{algebraic tensor product} 
\begin{align*}
	&B(\H_A)\otimes_{\text{alg}} B(\H_B) \\
	:=&
	\Big\{
	\sum_{i=1}^n a_i \otimes b_i,
	\,\Big|\,
	n \in \NN,
		a_i \in B(\H_{A}),
		b_i \in B(\H_{B})
	\Big\}\\
	\subset& B(\H_A\otimes H_B).
\end{align*}
Based on the situation in finite dimensions, we would expect equality between
$B(\H_A)\otimes_{\text{alg}} B(\H_B)$ and 
$B(\H_A\otimes H_B)$.
This behavior is recovered for general Hilbert spaces only after taking the closure of the algebraic tensor product with respect to the weak operator topology.
In general, if $\A\subset B(\H_A), \B\subset B(\H_B)$, then the weak operator closure of the algebraic tensor product (also known as the \emph{von~Neumann tensor product}
\cite[Sec.~III.1.5]{blackadar2006operator}) 
is denoted by $\A\bar\otimes\B\subset B(\H_A\otimes\H_B)$.

\section{Quantum correlations} \label{sec:Quantum correlations}

\subsection{Two-party quantum correlations} \label{sec:bell}

As a warm-up, we can now state precisely the two well-known distinct models of \emph{two-party quantum correlations}, i.e.\ the set of conditional probabilities $p(\alpha\beta|xy)$ obtainable by two parties performing local measurements on a shared quantum state.

\begin{definition}[Two-party correlations, tensor product model]
	\label{def:bipartite tp}
	A set $p(\alpha\beta|xy)$ of conditional probabilities is a \emph{bipartite quantum distribution in the Hilbert space tensor product model} if the following holds.
	There are
	\begin{itemize}
	\item
		Hilbert spaces $\H_A, \H_B$,
	\item
		for each of Alice's settings $x$ a POVM $\{A_{\alpha|x}\}_\alpha\subset B(\H_A)$,
		and 
		for each of Bob's settings $y$ a POVM $\{B_{\beta|y}\}_\beta\subset B(\H_B)$, 
	\item
		a density operator $\rho$ on $\H_A \otimes \H_B$
	\end{itemize}
	such that
	\begin{align*}
		p(\alpha\beta|xy) = \tr(\rho \, A_{\alpha|x}\otimes B_{\beta|y}).
	\end{align*}
\end{definition}

In order to highlight the essential difference, we first give a version of the commuting observables model that is phrased as closely as possible to the tensor product model.

\begin{definition}[Two-party quantum correlations, commuting observables model]
	\label{def:bipartite commuting}
	A set $p(\alpha\beta|xy)$ of conditional probabilities is a \emph{bipartite quantum distribution in the commuting observables model} if the following holds.
	There is 
	\begin{itemize}
	\item
		a Hilbert space $\H$
	\item
		for each of Alice's settings $x$ a POVM $\{A_{\alpha|x}\}_\alpha\subset B(\H)$,
		and 
		for each of Bob's settings $y$ a POVM $\{B_{\beta|y}\}_\beta\subset B(\H)$, 
		such that all of Alice's operators commute with all of Bob's,
	\item
		a density operator $\rho$ on $\H$
	\end{itemize}
	such that
	\begin{align*}
		p(\alpha\beta|xy) = \tr(\rho \, A_{\alpha|x} B_{\beta|y}).
	\end{align*}
\end{definition}

As alluded to before, \emph{Tsirelson's problem} asked whether the two definitions characterize the same set of correlations \cite{tsirelson2007open,scholz2008tsirelson,junge2011connes,fritz2012tsirelson}.
This problem has since been answered in the negative \cite{ji2020mip}.

There is an equivalent way of characterizing the commuting observable model. 
This version refers only to the observable algebras and not directly to any Hilbert space:

\begin{definition}[Two-party quantum correlations, commuting operator model: algebraic formulation]
	\label{def:bipartite commuting algebraic}
	A set $p(\alpha\beta|xy)$ of conditional probabilities is a bipartite quantum distribution \emph{in the commuting observable model} if the following holds.
	There are 
	\begin{itemize}
	\item
		a $C^*$-algebra $\D$ 
		(of global observables),
	\item
		two mutually commuting $C^*$-subalgebras $\A, \B\subset \D$ 
		(the observables measurable by the respective parties),
	\item
		for each of Alice's settings $x$ a POVM $\{A_{\alpha|x}\}_\alpha\subset \A$,
		and 
		for each of Bob's settings $y$ a POVM $\{B_{\beta|y}\}_\beta\subset \B$, 
	\item
		a state $\rho$ on $\D$
	\end{itemize}
	such that
	\begin{align*}
		p(\alpha\beta|xy) = \rho (A_{\alpha|x} B_{\beta|y}).
	\end{align*}
\end{definition}

Proving the equivalence between Def.~\ref{def:bipartite commuting} and Def.~\ref{def:bipartite commuting algebraic} amounts to an application of the GNS construction \cite{blackadar2006operator}.

\subsection{Bilocal correlations} \label{sec:bilocal correlations}

\begin{figure}
\centering
\includegraphics[width=.7\linewidth]{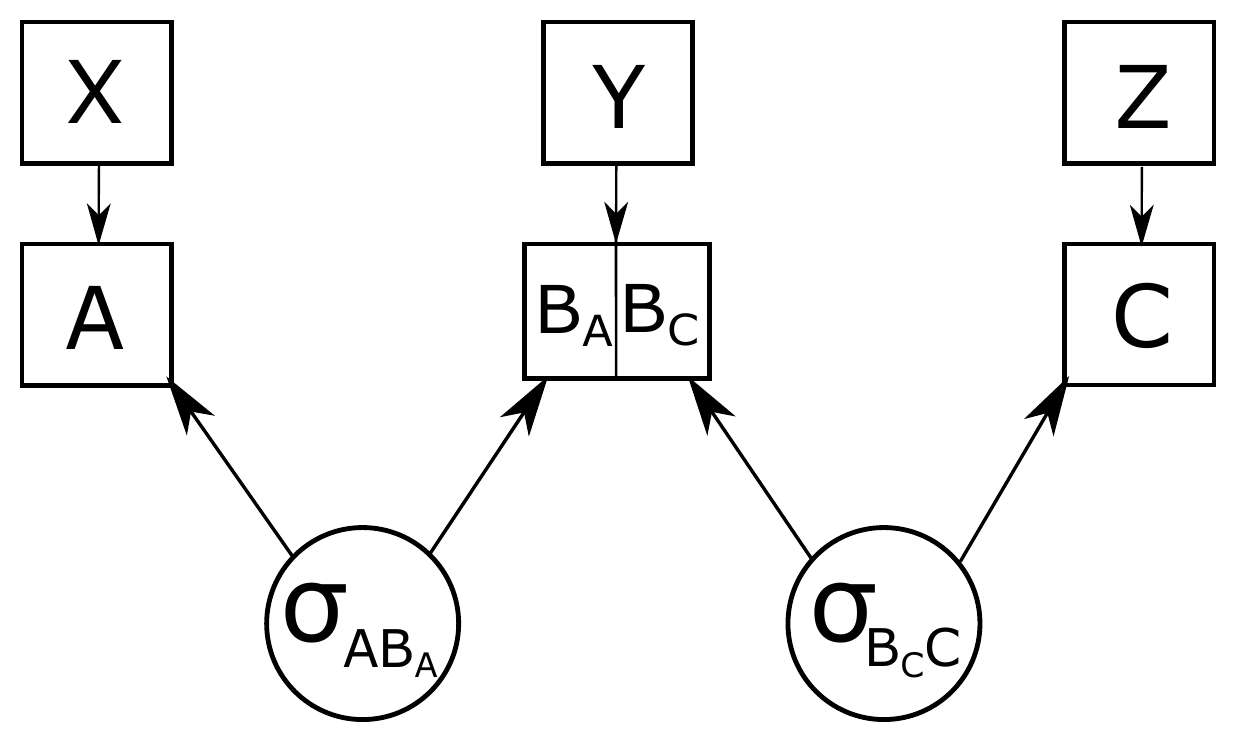} 
\caption{The bilocal scenario. Alice and Bob share a bipartite quantum state $\sigma_{A B_A}$ and Bob and Charlie share a bipartite quantum state $\sigma_{B_C C}$. 
Alice performs a measurement with the POVM $\{A_{\alpha|x}\}_\alpha$ based on the setting measurement setting $x$. 
Bob and Charlie perform a similar measurement.
The conditional probabilities $p(\alpha\beta\gamma|xyz)$ that can arise in this way are called bilocal correlations.
}
\label{fig:bilocal_scenario}
\end{figure}

When modeling locality using tensor products of Hilbert spaces, the set of bilocal correlations is defined as follows.

\begin{definition}[Tensor product model]\label{def:tensor}
	A set $p(\alpha\beta\gamma|xyz)$ of conditional probabilities is a bilocal quantum distribution \emph{in the product tensor model} if the following holds.
	There are
	\begin{itemize}
	\item
		Hilbert spaces $\H_A, \H_{B_A}, \H_{B_C}, \H_C$,
	\item
		for each of the settings $x,y,z$ POVMs 
		\begin{align*}
			\{A_{\alpha|x}\}_\alpha &\subset B(\H_A), \\
			\{B_{\beta|y}\}_\beta &\subset 
			B(\H_{B_A})\bar\otimes B( \H_{B_C}), \\
			\{C_{\gamma|z}\}_\gamma &\subset B(\H_C),
		\end{align*}
	\item
		density operators 
		$\sigma_{A B_A}$ on $\H_A\otimes\H_{B_A}$
		and
		$\sigma_{B_C C}$ on $\H_{B_C}\otimes\H_{C}$,
	\end{itemize}
	such that
	\begin{align*}
		p(\alpha\beta\gamma|xyz) = \tr\big((\sigma_{A B_A}\otimes \sigma_{B_C C})\,(A_{\alpha|x}\otimes B_{\beta|y} \otimes C_{\gamma|z})\big).
	\end{align*}
\end{definition}

In the commuting observables-model, bilocality takes on the following form:

\begin{definition}[Commuting observables model]\label{def:commuting}
	A set $p(\alpha\beta\gamma|xyz)$ of conditional probabilities is said to be a \emph{bilocal quantum distribution in the commuting observables model} if the following holds.
	There are
	\begin{itemize}
		\item 
			a $C^*$-algebra $\D$,
		\item
			mutually commuting $C^*$-subalgebras 
			$\A, \B_A, \B_C, \C\subset \D$,
	\item
		for each of the settings $x,y,z$ POVMs 
		\begin{align*}
			\{A_{\alpha|x}\}_\alpha &\subset \A, \\
			\{B_{\beta|y}\}_\beta &\subset  
			\B_A \cdot \B_C, \\
			\{C_{\gamma|z}\}_\gamma &\subset \C,
		\end{align*}
		with $\B:=\B_A \cdot \B_C$ the subalgebra of $\D$ generated by $\B_A$ and $\B_C$,
	\item
		a state $\rho$ on $\D$ that acts as a product state in the sense
		\begin{align}\label{eqn:covert}
			\rho(a b_A b_C c) = \rho(a b_A)\rho(b_C c)
		\end{align}
		for all $a\in\A,
			b_A\in\B_A,
			b_C\in\B_C,
			c\in\C$,
	\end{itemize}
	such that
	\begin{align*}
		p(\alpha\beta\gamma|xyz) = \rho (A_{\alpha|x}\, B_{\beta|y} \,  C_{\gamma|z}).
	\end{align*}
\end{definition}

\subsection{Equivalent characterizations of bilocal quantum correlations} \label{sec:equivalences}

We will give three further characterizations of the set of bilocal quantum correlations in the commuting operator model.
Some of these equivalences are integral to our completeness proof -- but they might also be of independent interest.

All statements made here are corollaries of the technical Theorem~\ref{thm:fact to biloc} proven in Sec.~\ref{sec:fac to biloc}.

\subsubsection{A reduced factorization condition}

To motivate the first reformulation in Corollary \ref{cor:reduced} below, let us try to see which part of Def.~\ref{def:commuting} might be the most difficult to work with algorithmically.
In our assessment, this is the ``hidden factorization condition'' of Eq.~(\ref{eqn:covert}).
It is ``hidden'' in the sense that it involves product operators $b_A b_C, b_A \in \B_A, b_C \in \B_C$ that \emph{need not lie in the algebra generated by the POVMs}.
But it is properties of precisely this algebra that methods building on the non-commutative polynomial optimization (NPO) hierarchy \cite{pironio2010convergent}, 
used in e.g.\ the quantum inflation method of \cite{wolfe2019inflation}, 
typically optimize over.  
As argued in more detail in Ref.~\cite[Sec. 2.5]{ligthart2021convergent}, this poses a barrier against proving completeness for such methods, including the original quantum inflation scheme.

Reference~\cite{ligthart2021convergent} circumvents this problem by explicitly adding generators for the algebras $\B_A, \B_C$ to the input of the NPO hierarchy, and expressing the POVM elements as finite-rank superpositions of those.
The price to pay for this workaround consists of additional computational costs, as well as the necessity to upper-bound this ``Schmidt rank'' of the POVM elements.

The following corollary shows that in the special case of the bilocal scenario, these difficulties can fortunately be avoided.
Indeed, the weaker factorization condition (\ref{eqn:overt}), 
involving only operators generated by the measured POVMs,
suffices to imply the \emph{a priori} more general (\ref{eqn:covert}).
We will refer to the weaker constraints as the \emph{reduced model}.

\begin{corollary}[Reduced model]\label{cor:reduced}
	A set $p(\alpha\beta\gamma|xyz)$ of conditional probabilities is bilocal in the commuting observables model of Def.~\ref{def:commuting} if and only if 
there are
	\begin{itemize}
		\item 
			a $C^*$-algebra $\D$,
		\item
			mutually commuting $C^*$-subalgebras 
			$\A, \B, \C\subset \D$,
	\item
		for each of the settings $x,y,z$ POVMs 
		\begin{align*}
			\{A_{\alpha|x}\}_\alpha &\subset \A, \\
			\{B_{\beta|y}\}_\beta &\subset  \B, \\
			\{C_{\gamma|z}\}_\gamma &\subset \C,
		\end{align*}
	\item
		a state $\rho$ on $\D$ that acts as a product state in the sense
		\begin{align}\label{eqn:overt}
			\rho(a  c) = \rho(a )\rho( c)
		\end{align}
		for all $a\in\A,
			c\in\C$,
	\end{itemize}
	such that
	\begin{align*}
		p(\alpha\beta\gamma|xyz) = \rho (A_{\alpha|x}\, B_{\beta|y} \,  C_{\gamma|z}).
	\end{align*}
\end{corollary}

We note that the reduced model arises implicitly from the \emph{factorisation bilocal NPA hierarchy} of \cite{renou2022two}.

\subsubsection{Bilocal Tsirelson Problem}

The second reformulation, specified in Corollary~\ref{cor:mixed} below, clarifies the differences between the Hilbert space tensor product model and the commuting operator model of bilocal correlations.

The two approaches are obviously different:
To see this, one can simply embed the two-party  scenario into the bilocal one,
e.g.\ by taking the $A$ system or $C$ system to be trivial.

It could be surmised that there are ``genuine bilocal differences'' between the two approaches, and that the bilocal scenario could teach us about Tsirelson's Problem in a way that goes beyond the two-party case.
We will, however, show that this is not the case.
More precisely, consider the \emph{mixed model} formalized below, where the bipartition $AB_A|B_CC$ is described by a Hilbert space tensor product, while all we can say about the bipartitions $A|B_A$ and $B_C|C$ is that they are associated with commuting observable algebras.

\begin{corollary}[Mixed model]
	\label{cor:mixed}
	A set $p(\alpha\beta\gamma|xyz)$ of conditional probabilities is bilocal in the commuting observables model of Def.~\ref{def:commuting} if and only if 
	there are
	\begin{itemize}
		\item
		Hilbert spaces $\H_{A B_A}, \H_{B_C C}$,
		\item
		mutually commuting 
	  $C^*$-algebras	
		\begin{align*}
			 \A,  \B_A \subset B(\H_{A B_A}), 
			\qquad 
			 \B_C,  \C \subset B(\H_{B_C C}),
		\end{align*}
	  \item
		for each of the settings $x,y,z$ POVMs
		\begin{align*}
			\{ A_{\alpha|x} \}_\alpha &\subset  \A, \\
			\{ B_{\beta|y} \}_\beta &\subset  \B_A\bar\otimes\B_C, \\
			\{ C_{\gamma|z} \}_\gamma &\subset  \C,
		\end{align*}
		\item
			density operators 
			$\sigma_{A B_A}$ on $\H_{A B_A}$ and
			$\sigma_{B_C C}$ on $\H_{B_C C}$,
	\end{itemize}
	such that
	\begin{align*}
		p(\alpha\beta\gamma|xyz)
		&=
		\tr\big((\sigma_{A B_A}\otimes \sigma_{B_C C})\,(A_{\alpha|x}\, B_{\beta|y} \, C_{\gamma|z})\big),
	\end{align*}
	where all operators act on $B(\H_{A B_A} \otimes\H_{B_C C})$ in the natural way.
\end{corollary}

\subsubsection{The Renou-Xu formulation}

Finally, we consider the formulation used in Ref.~\cite{renou2022two}.
It could be described as the bilocal analogue of Def.~\ref{def:bipartite commuting}, in the sense that it formalizes a commuting observables-model while avoiding to explicitly introduce the local observable algebras.
While in the two-party  case, the equivalence of Def.~\ref{def:bipartite commuting} and Def.~\ref{def:bipartite commuting algebraic} was a direct consequence of the GNS construction, the relation between the Renou-Xu model and commuting operator models defined above may not be as obvious.
However, we will show:

\begin{corollary}[Renou-Xu model \cite{renou2022two}]
	\label{cor:renou-xu}
	A set $p(\alpha\beta\gamma|xyz)$ of conditional probabilities is bilocal in the commuting observables model of Def.~\ref{def:commuting} if and only if 
	there are 
	\begin{itemize}
		\item
		a Hilbert space $\H$,
		\item
			commuting projection operators $P, Q\in B(\H)$, such that $PQ$ is a normalized rank-one projection,
	  \item
			for each of Alice's settings $x$, a 
			POVM $\{A_{\alpha|x}\}_\alpha\subset B(\H)$,
			and likewise for Bob and Charlie, such that:
			(1) operators belonging to different parties commute,
			and (2)
		\begin{align*}
			[ A_{\alpha|x}, Q]=
			[P,  C_{\gamma|z}]=0,
		\end{align*}
	\end{itemize}
	such that
	\begin{align*}
		p(\alpha\beta\gamma|xyz)
		&=
		\tr\big( PQ
		\,  A_{\alpha|x}  B_{\beta|y}  C_{\gamma|z} \big).
	\end{align*}
\end{corollary}

\subsection{The finite-dimensional case}

In the two-party case, the distinction between the tensor product model and the commuting observable model ceases to exist if either can be realized in finite dimensions.
Reference~\cite{renou2022two} asked whether the same is true for the bilocal scenario.
Here, we answer this question in the affirmative.
In fact, the equivalence already holds when both Alice and Charlie can be associated with a finite-dimensional system.

\begin{corollary}
\label{cor:finite}
	Assume $p(\alpha\beta\gamma|xyz)$ is compatible with \emph{any} of the models given in
	Def.~\ref{def:tensor}, 
	Def.~\ref{def:commuting}, 
	Cor.~\ref{cor:reduced}, 
	Cor.~\ref{cor:mixed}, 
	Cor.~\ref{cor:renou-xu},
	and is such that the $C^*$-algebra 
	generated by Alice's and Charlie's POVMs are finite-dimensional. 

	Then $p(\alpha\beta\gamma|xyz)$ is compatible with \emph{all} these models, and all operator algebras and Hilbert spaces can be chosen to be finite-dimensional.
\end{corollary}

\subsection{Proof of the equivalences}
\label{sec:fac to biloc}

The claimed equivalences derive from the following theorem.
We state it in general terms (i.e.\ not yet specific to the various models of bilocality).

Recall 
\cite[Sec.~II.6.4]{blackadar2006operator}.
that the GNS construction associates with every $C^*$-algebra $\F$ and state $\sigma\in K(\F)$ a  triple $(\H, \pi, \ket{\Omega})$, where $\H$ is a Hilbert space, $\pi: \F\to B(\H)$ a $*$-representation, 
and $\ket{\Omega}\in\H$ a cyclic vector that implements the state in the sense that $\sigma(f)=\langle\Omega|\pi(f)|\Omega\rangle$ for every $f\in\F$.

\begin{theorem}
	\label{thm:fact to biloc}
	Let ${\A}, {\B}, {\C}$ be mutually commuting
 	$C^*$-subalgebras
	of some 
	C$^*$-algebra ${\D}$.
	Let $\rho$ be a state on ${{\D}}$ such that
	\begin{align}\label{eqn:factorization}
		\rho( a  c) = \rho(a )\rho( c)
		\qquad\forall a\in {\A}, c\in {\C}.
	\end{align}

	Let $(\H_{{A B_A}}, \pi_{{A}}, |\Omega_{{A}}\rangle)$ be the GNS representation of ${\A}$ associated with the state $\rho$.
	Let ${{\B}}_{{A}}$ be the commutant of $\pi_{{A}}({\A})$ in $B(\H_{AB_A})$. 
	Define $(\H_{B_C C}, \pi_{{C}}, |\Omega_{{C}}\rangle)$ and 
	${{\B}}_{{C}}\subset B(\H_{B_C C})$ analogously.

	Then there exists a completely positive unital map
	\begin{align*}
		\Lambda: {\B} \to {{\B}}_{{A}} \bar\otimes {{\B}}_{{C}}\subset B(\H_{A B_A}\otimes \H_{B_C C})
	\end{align*}
	such that
	for all $a\in{\A}, b\in{\B}, c\in{\C}$ 
	\begin{align*}
		\rho(abc)
		&=
		\tr\big( |\Omega_{{A}}\rangle\langle\Omega_{{A}}|\otimes|\Omega_{{C}}\rangle\langle\Omega_{{C}}|\,
		\, \pi_{{A}}(a) \Lambda(b) \pi_{{C}}(c) \big),
	\end{align*}
	where all operators act on $B(\H_{A B_A} \otimes\H_{B_C C})$ in the natural way.
\end{theorem}

The spaces $\H_{A B_A}, \H_{B_C C}$ have previously appeared in the proof of Thm.~3.2 in Ref.~\cite{renou2022two} (as $V_{A B_L}, V_{B_R C}$).
In fact, this inspired our formulation of Thm.~\ref{thm:fact to biloc}.
We go beyond this prior result by showing that they give rise to a tensor product structure on the global Hilbert space.

To prove the theorem, consider in addition the GNS representation
$(\pi_{{\D}}, \H_{{\D}}, |\Omega_{{\D}}\rangle)$ of ${{\D}}$ associated with $\rho$.

\begin{lemma}\label{lem:tensor embedding}
	There is an isometric embedding $V: \H_{A B_A}\otimes\H_{B_C C}\to\H_{{\D}}$ which fulfills
	\begin{align}
		V\,|\Omega_{{A}}\rangle \otimes |\Omega_{{C}}\rangle 
		&=|\Omega_{{\D}}\rangle, \label{eqn:isometric omegas} \\
		V\pi_{{A}}(a)\otimes\pi_{{C}}(c) &= \pi_{{\D}}(ac) V. \label{eqn:intertwiner}
	\end{align}
\end{lemma}

Equation~(\ref{eqn:intertwiner}) says that $V$ intertwines $\pi_A\otimes\pi_C$ and $\pi_{\D}$ as representations of the $C^*$-algebra generated by $\A$ and $\C$.

\begin{proof}
	The factorization property (\ref{eqn:factorization}) implies that $|\Omega_{{A}}\rangle\otimes|\Omega_{{C}}\rangle$ and $|\Omega_{{\D}}\rangle$ induce the same state on the C$^*$-algebra generated by $\A$ and $\C$:
	\begin{align*}
		&\langle\Omega_{{A}}|\otimes\langle\Omega_{{C}}|\,(\pi_{{A}}(a)\otimes\pi_{{C}}(c))\,|\Omega_{{A}}\rangle \otimes |\Omega_{{C}}\rangle  \\
		=&\rho(a)\rho(c) = \rho(ac) 
		=\langle \Omega_{{\D}}| \pi_{{\D}}(ac)|\Omega_{{\D}}\rangle.
	\end{align*}
	By the uniqueness property of the GNS construction \cite[Proposition~4.5.3]{kadison1983fundamentals}, there exists a unitary 
	\begin{align*}
		V: \H_{A B_A}\otimes\H_{B_C C}\to \overline{\pi_{{\D}}({\A}{\C})|\Omega_{{\D}}\rangle} =: \K
	\end{align*}
	such that
	\begin{align*}
		V |\Omega_A\rangle\otimes|\Omega_C\rangle &= |\Omega_\D\rangle, \\
		V \pi_{{A}}(a)\otimes\pi_{{C}}(c) V^* 
																							&=  \pi_{\D}(a c) \upharpoonright\K,
	\end{align*}
	where the final symbol denotes the restriction of $\pi_{\D}$ to $\K=\range V$.
	The advertised intertwining relation follows by multiplying the last line with $V$ from the right and finally re-interpreting $V$ as a map to all of $\H_{\D}$.
\end{proof}

\begin{lemma}
	It holds that
	\begin{align*}
		V^*\pi_{{\D}}({{\B}}) V \subset 
	{{\B}}_{{A}} \bar\otimes {{\B}}_{{C}}.
	\end{align*}
\end{lemma}

\begin{proof}
	We first claim that 
	\begin{align*}
		V^*\pi_{{\D}}({{\B}}) V\subset (\pi_{{A}}({{\A}})\otimes_{\text{alg}}\pi_{{C}}({{\C}}))'.
	\end{align*}
	Indeed, for $a\in{{\A}}, b\in{{\B}}, c\in{{\C}}$, 
	Eq.~(\ref{eqn:intertwiner}) and its adjoint give
	\begin{align*}
		& [V^* \pi_{{\D}}(b) V, \pi_{A}(a)\otimes\pi_{C}(c)] \\
		=& 
		V^* \pi_{{\D}}(b) V
		\pi_{A}(a)\otimes\pi_{C}(c) 
		-
		\pi_{A}(a)\otimes\pi_{C}(c) 
		V^* \pi_{{\D}}(b) V
		\\
		=& 
		V^* \pi_{{\D}}(b) \pi_{\D}(ac) 
		V
		-
		V^* 
		\pi_{\D}(ac) 
		\pi_{{\D}}(b) V
		\\
		=& 
		V^* [\pi_{{\D}}(b),\pi_{\D}(ac) ] V 
		= 0.
	\end{align*}
	Now use the fact that the commutator of a set equals the commutator of its weak operator closure \cite[I.2.5.3]{blackadar2006operator},
	the Bicommutant Theorem \cite[Theorem~5.3.1]{kadison1983fundamentals},
	and the Commutation Theorem for von~Neumann algebras \cite[III.4.5.8]{blackadar2006operator} 
	to conclude
	\begin{align*}
		(\pi_{A}({{\A}})\otimes_{\text{alg}}\pi_{C}({{\C}}))' 
		&=
		(\pi_{A}({{\A}})\bar\otimes\pi_{C}({{\C}}))' \\ 
		&=
		(\pi_{A}({\A})''\bar\otimes\pi_{C}({\C})'')' \nonumber \\
		&=
		(\pi_{A}({\A})'''\bar\otimes\pi_{C}({\C})''')  \\ 
		&=
		(\pi_{A}({\A})'\bar\otimes\pi_{C}({\C})'). \nonumber
	\end{align*}
\end{proof}

\begin{proof}[Proof (of Theorem~\ref{thm:fact to biloc})]
	Set
	\begin{align*}
		\Lambda: b \mapsto V^* \pi_{{\D}}(b) V
	\end{align*}
	and compute, using Lemma~\ref{lem:tensor embedding} repeatedly,
	\begin{align*}
		&\langle\Omega_{A}|\langle\Omega_{C}| \, \pi_{A}(a) (V^*\pi_{{\D}}(b)V) \pi_{C}(c) |\Omega_{A}\rangle |\Omega_{C}\rangle \\
		=&
		\langle\Omega_{A}|\langle\Omega_{C}| \, V^*\pi_{\D}(a) \pi_{{\D}}(b) \pi_{\D}(c) V |\Omega_{A}\rangle |\Omega_{C}\rangle \\
		=&
		\langle\Omega_{\D}| \, \pi_{{\D}}(abc) |\Omega_{\D}\rangle = \rho(abc).
	\end{align*}
\end{proof}

After these preparations, we can now proceed to prove the equivalences claimed to hold in Sec.~\ref{sec:equivalences}.
The proof's chain of implications among the various models is visualized in Fig.~\ref{fig:implications}.
\begin{figure*} 
	\begin{align*}
		\left.
		\begin{array}{l}
		\text{tensor product} \\ 
	\text{commuting observables} \\ 
			\text{Renou-Xu} 
		\end{array}
		\right\}
		\Rightarrow
		\text{reduced}  
		\underset{\text{Thm.~\ref{thm:fact to biloc}}}{\Longrightarrow}
	\text{mixed} 
		\Rightarrow
		\left\{
		\begin{array}{l}
			\text{[if $\dim\A, \dim\C <\infty$] tensor product} \\
			\text{commuting observables} \\
		\text{Renou-Xu} 
		\end{array}
		\right.
	\end{align*}
	\caption{
		\label{fig:implications}
		Logical structure of the proof given in Sec.~\ref{sec:fac to biloc}.
		The equivalences claimed in Sec.~\ref{sec:equivalences} follow from this chain of implications among the various models of bilocal quantum correlations.
	}
\end{figure*}
The implication ``reduced model $\Rightarrow$ Renou-Xu model'' also follows from Thm.~3.2 of \cite{renou2022two}.

\begin{proof}[Proof (of the equivalences stated in Sec.~\ref{sec:equivalences})]
	\ \\ \textbf{Step 1}:
	We claim that if $p(\alpha\beta\gamma|xyz)$ is compatible with 
	the tensor product model of Def.~\ref{def:tensor}, the commuting observables model of Def.~\ref{def:commuting}, or the Renou-Xu model of Cor.~\ref{cor:renou-xu},
	then it is also compatible with the reduced model (Cor.~\ref{cor:reduced}).
	
	This is straightforward to verify, except perhaps for the Renou-Xu model, which we treat explicitly.
	Indeed, consider a Renou-Xu model realization for $p(\alpha\beta\gamma|xyz)$.
	Let $\A$ be the $C^*$-algebra generated by Alice's POVM elements, and likewise for Bob and Charlie.
	Let $\D$ be the $C^*$-algebra generated by $\A, \B, \C$.
	By assumption, there is a normalized vector $\ket\psi\in\H$, such that $PQ=|\psi\rangle\langle\psi|$.
	Let $\rho$ be the associated vector state $\rho(x)=\langle\psi|x|\psi\rangle$ on $\D$.
	We have now constructed all objects that enter the reduced model.
	Using the commutation relations of the Renou-Xu model, one verifies the factorization constraint (\ref{eqn:overt}) for all $a\in\A, c\in\C$:
	\begin{align*}
		\rho(ac)
		&=
		\tr(PQ a c) \\
		&=
		\tr(PPQQ ac) \\
		&=
		\tr(PQ a QP c) \\
		&=
		\langle\psi| a|\psi\rangle\langle\psi| c|\psi\rangle \\
		&=\rho(a)\rho(c)
	\end{align*}
	(a similar calculation appears in the proof of Thm.~3.2 of Ref.~\cite{renou2022two}).

	\textbf{Step 2}
	is to show that, 
	by Thm.~\ref{thm:fact to biloc}, 
	the reduced model implies the mixed model.
	Assume a reduced model description with elements $\A,\B,\C,\D,\rho, A_{\alpha|x}, \dots$ is given.
	They satisfy the assumptions of Thm.~\ref{thm:fact to biloc}, so that we can use the objects whose existence it guarantees in the construction of the mixed model.
	Indeed, one immediately verifies the properties of the mixed model from the choices
	\begin{alignat*}{3}
		\H_{AB_A}^{(\text{mix})} &= 
		\H_{AB_A}, &
		\H_{B_CC}^{(\text{mix})} &= \H_{B_C C}, \\
		\A^{(\text{mix})} &= \pi_{A}(\A),  &
		\C^{(\text{mix})} &= \pi_{C}(\C), \\
		\B_A^{(\text{mix})} &= \B_A,  &
		\B_C^{(\text{mix})} &= \B_C, \\ 
		A^{(\text{mix})}_{a|x}  &= \pi_A(A_{\alpha|x}), & 
		B^{(\text{mix})}_{b|y}  &= \Lambda(B_{\beta|y}), & 
		C^{(\text{mix})}_{c|z}  &= \pi_C(C_{\gamma|z}), \\
		\sigma^{(\text{mix})}_{A B_A} &= |\Omega_A\rangle\langle\Omega_A|, &
		\sigma^{(\text{mix})}_{B_C C} &= |\Omega_C\rangle\langle\Omega_C|.
	\end{alignat*}

	\textbf{Step 3}:
	The mixed model obviously implies the commuting operator model.
	It also implies the Renou-Xu model by setting 
	\begin{align*}
		P = |\Omega_{A}\rangle\langle\Omega_{A}|\otimes \Id_{\B_C \C}, \quad 
		Q = \Id_{\A \B_A}\otimes|\Omega_{C}\rangle\langle\Omega_{C}|
	\end{align*}
	(which is similar to the construction of their 
	operators $\rho, \sigma$ in the 
	proof of Thm.~3.2 of Ref.~\cite{renou2022two}).
\end{proof}

It remains to treat the finite-dimensional case, as advertised in Cor.~\ref{cor:finite}.
The proof combines the construction given in Thm.~\ref{thm:fact to biloc} with the well-known fact
(c.f.\ Ref.~\cite{tsirelson2007open,scholz2008tsirelson})
that for two-party correlations, a finite-dimensional commuting model (as in Def.~\ref{def:bipartite commuting}) implies a tensor product model (as in Def.~\ref{def:bipartite tp}).
Specifically, we will use the following reformulation of  Theorem~1 of Ref.~\cite{scholz2008tsirelson}:

\begin{lemma}
	\label{lem:tsirelson finite}
	If $\F,\G$ are mutually commuting $C^*$-algebras on a finite-dimensional Hilbert space $\H$, then there exist finite-dimensional Hilbert spaces $\H_\F, \H_\G$ and an isometric embedding 
	\begin{align*}
		W_{\F\G}: \H \to \H_\F\otimes\H_\G
	\end{align*}
  such that
	\begin{align*}
		W_{\F\G}\, \F\, W_{\F\G}^* &\subset B(\H_\F)\otimes \Id, \\
		W_{\F\G}\, \G\, W_{\F\G}^* &\subset \Id\otimes B(\H_\G).
	\end{align*}
\end{lemma}

In keeping with the notation of $C^*$-algebras, $W^*$ denotes the adjoint of $W$, i.e.\ the operator that would be denoted as $W^\dagger$ in physics notation.

\begin{proof}[Proof (of the equivalences for finite-dimensional models)]
	By the previous proof, all models imply the mixed model, where specifically $\H_{A B_A}, \H_{B_C C}$ arise from the GNS representation of the observable algebras $\A$ and $\C$ respectively.
	In particular, the dimensions of these Hilbert spaces are upper-bounded by the dimension of the associated algebras,
	which in turn can be chosen to be the ones generated by Alice's and Charlie's observables.
	All operators that enter the construction of the mixed model as laid out in Cor.~\ref{cor:mixed} are linear maps on the tensor product of these two Hilbert spaces and therefore finite-dimensional.
	The third step of the previous proof then gives finite-dimensional realizations in the commuting operator model and the Renou-Xu model.

	It remains to be shown that $p(\alpha\beta\gamma|xyz)$ can be realized in the Hilbert space tensor product model (Def.~\ref{def:tensor}), and in particular in one involving only finite-dimensional spaces.
	Let a finite-dimensional mixed model realization of $p(\alpha\beta\gamma|xyz)$
	with elements $\H_{A B_A}, \H_{B_C C}, \A, \B_A, \B_C, \B, A_{\alpha|x}, \sigma_{A B_A}, \dots$
	be given.
	Our strategy is to apply Lem.~\ref{lem:tsirelson finite} separately to $\A, \B_A$ and to $\B_C, \C$.
	First, choosing $\F=\A$ and $\G=\B_A$ in Lem.~\ref{lem:tsirelson finite} establishes the existence of two Hilbert spaces $\H_{\A}, \H_{\B_A}$ and an isometry 
	\begin{align*}
		W_{\A\B_A}: \H_{A B_\A} \to \H_{\A} \otimes \H_{\B_A}.
	\end{align*}
  These allow us to choose the first set of objects that will enter the tensor product model as
	\begin{align*}
		\H_A^{(\text{t.p.})}&=\H_{\A}, &
		\H_{B_A}^{(\text{t.p.})}&=\H_{\B_A}, \\
		A_{\alpha|x}^{(\text{t.p.})} &= W_{\A\B_A} A_{\alpha|x} W_{\A\B_A}^*, \\
		\sigma^{(\text{t.p.})}_{A B_A} &= W_{\A\B_A}\sigma_{A B_A} W_{\A\B_A}^*.
	\end{align*}
	An analogous procedure starting with $\B_C, \C$ gives
	\begin{align*}
		\H_{B_C}^{(\text{t.p.})}&=\H_{\B_C}, & \H_{C}^{(\text{t.p.})}&=\H_{\C}, \\
		C_{\gamma|z}^{(\text{t.p.})} &= W_{\B_C\C} C_{\gamma|z} W_{\B_C\C}^*, \\
		\sigma^{(\text{t.p.})}_{B_C C} &= W_{\B_C\C}\sigma_{B_C C} W_{\B_C\C}^*.
	\end{align*}
	Finally, with
	\begin{align*}
		B_{\beta|y}^{(\text{t.p.})} &= 
		(W_{\A\B_A}\otimes W_{\B_C\C} )
			B_{\beta|y} 
		(W_{\A\B_A}\otimes W_{\B_C\C})^*
	\end{align*}	
	it is straight-forward to verify the properties of the tensor product model.
\end{proof}

It is apparent from the proof that the condition in Cor.~\ref{cor:finite} can be slightly weakened.
Instead of demanding that the algebras $\A, \C$ be finite-dimensional, it is sufficient for the conclusions to hold that the GNS Hilbert space $\H_{AB_A}\otimes \H_{B_C C}$ associated with the restriction of the state to $\A \C$ is finite-dimensional.

\section{Two complete hierarchies for the reduced model} \label{sec:complete hierarchy}

In this section, we will construct complete hierarchies of relaxations for the reduced model defined in Cor.~\ref{cor:reduced}.
Most ingredients for this construction and the completeness proof have been developed in Ref.~\cite{ligthart2021convergent} (based on Refs.~\cite{pironio2010convergent, wolfe2019inflation, wolfe2021quantum, raggio1989quantum}), to which we will refer for technical details.

\subsection{Outline}

Let $\D$ be the \emph{universal $C^*$-algebra} 
(\cite[Sec.~II.8.3]{blackadar2006operator}
, \cite[Sec.~2.2]{ligthart2021convergent})
with generators 
\begin{align}\label{eqn:alt gens D}
	\G=\{\Id, A_{\alpha|x}, B_{\beta|y}, C_{\gamma|z}\}
\end{align}
and relations
\begin{align}\label{eqn:alt rels}
	& [A_{\alpha|x}, B_{\beta|y}] = 0, \qquad &\forall \alpha,\beta,x,y, \\
	& [B_{\beta|y}, C_{\gamma|z}] = 0, \qquad &\forall \beta,\gamma,y,z, \\
	& [C_{\gamma|z}, A_{\alpha|x}] = 0, \qquad &\forall \alpha,\gamma,x,z, \\
	& [\Id, X] = 0 \qquad &\forall X\in \G, \\
	& \Id X = X \Id = X \qquad &\forall X \in \G, \\
& X^* = X \succeq 0, \qquad &\forall X \in \G, \\
& \sum_\alpha A_{\alpha|x} = \sum_\beta B_{\beta|y} = \sum_\gamma C_{\gamma|z} = \Id & \forall x,y,z. \label{eqn:alt rels end}
\end{align}
In a precise sense \cite{blackadar2006operator,ligthart2021convergent}, 
$\D$
is the direct sum of all possible realizations of this algebra as operators on Hilbert spaces.
Let $K(\D)$ be the set of all states on $\D$.

We aim to solve the following optimization problem: 
\begin{align}
	\begin{split} \label{eq:bilocal_opt}
		f^* = \min_{\rho \in K(\D)} \ & \sum_{\alpha, \beta, \gamma, x,y,z} \left( \rho(A_{\alpha|x} B_{\beta|y} C_{\gamma|z} ) - p(\alpha \beta \gamma|x y z) \right)^2\\
		\text{s. t. } \quad
		& \rho(ac) - \rho(a) \rho(c) = 0  \qquad \forall a \in \A,\ c \in \C
	\end{split}
\end{align}
The objective function of this problem represents the minimal 2-norm distance between a quantum realization of the reduced model of the bilocal scenario and the observed statistics. 
We accept that the correlations can arise from a reduced model if $f^* = 0$ (or at most some small $\varepsilon$ that represents numerical and statistical tolerances).
By Theorem \ref{thm:fact to biloc} this means that such correlations can also arise in the mixed model, the commuting observables model and the Renou-Xu model of the bilocal scenario.
If $f^* > 0$, the correlations cannot have been produced in any of the models, including the tensor product one.

The problem \eqref{eq:bilocal_opt} is ``polynomial'' in two different ways:
The operators 
$A_{\alpha|x} B_{\beta|y} C_{\gamma|z}$ 
and $ac$ are (norm limits of) non-commutative polynomials in the generators,
while the objective function and the constraints are second order polynomials in the state.

As described in Ref.~\cite{ligthart2021convergent}, we will use two different techniques to deal with these non-linearities:
\begin{enumerate}
	\item
	Non-commutative polynomial optimization (NPO) \cite{pironio2010convergent} provides a hierarchy of SDP relaxations for optimizing over \emph{linear} functions on states of the universal algebra $\D$,
	subject to linear constraints.
	Its completeness follows from the GNS construction.
	\item
		By passing to their \emph{polarizations}, one can interpret the polynomial functions on $K(\D)$ as linear functions on symmetric product states on multiple copies of $\D$.
	Such states are constructed from symmetric extensions of states on $\D$ and completeness follows from a suitable quantum de~Finetti theorem \cite{ligthart2021convergent}.
\end{enumerate}

In Sec.~\ref{sec:polarization} below, we lay out how to use the results of Ref.~\cite{ligthart2021convergent} to construct a converging hierarchy of SDP relaxations for polynomial optimization problems over algebras given in terms of generators and relations.
While we focus on the bilocal scenario, the techniques can be straightforwardly adapted to general algebras and polynomials.
We call this approach the \emph{polarization hierarchy}.
In Sec.~\ref{sec:inflation hierarchy}, we describe a slightly different approach more closely related to quantum inflation \cite{wolfe2021quantum}, which we also prove to be complete.

\subsection{Polarization hierarchy}
\label{sec:polarization}

To define the polarization hierarchy, choose a \emph{level} $n\in\NN$ and consider $n$ copies of the generators:
\begin{align}\label{eqn:alt gens Dn}
	\G^{n}=\{\Id, A_{\alpha|x}^{(i)}, B_{\beta|y}^{(i)}, C_{\gamma|z}^{(i)}\} \qquad i \in 1, \dots n.
\end{align}
Relations analogous to those in Eqs.~(\ref{eqn:alt rels})-(\ref{eqn:alt rels end}) are imposed for each $i$, together with relations stating that operators for different values of the superscript $i$ commute.
The resulting universal $C^*$-algebra $\D^n$ is the ``largest $C^*$-algebra generated from $n$ commuting copies of $\D$'', or, more precisely, the \emph{maximal $C^*$-tensor product} $\D^n = \D^{\otimes_{\max} n}$ \cite{bratteli2012operator, ligthart2021convergent}.

We note that this algebra is closely related to the algebra that is constructed for the most general version of the quantum inflation technique \cite{wolfe2021quantum}. 
This technique works with an even larger algebra, where e.g.\ Bob's operators carry two indices $B_{\beta|y}^{(i,j)}$ that can be varied independently.
It will turn out that our simpler model is sufficient for the bilocal scenario.

On the $n$-th tensor product of $\D$, we can linearize $n$-th order polynomial functions on $K(\D)$ by passing to their polarization as follows \cite[Sec.~4.1.2]{ligthart2021convergent}:
With every state $\sigma\in K(\D)$ associate its $n$-fold symmetric product state $\Pi_\sigma^{n} \in K(\D^n)$ which is defined by its action on product operators in the obvious way:
\begin{align*}
	\Pi_\sigma^{n}(x_1\otimes\dots\otimes x_n) = \sigma(x_1) \dots \sigma(x_n)
\end{align*}
and extended to all of $\D^n$ by linearity and continuity.
Then
\begin{align*}
	&\sum_{\alpha, \beta, \gamma, x,y,z} \left( \sigma(A^{(1)}_{\alpha|x} B^{(1)}_{\beta|y} C^{(1)}_{\gamma|z} ) - p(\alpha \beta \gamma|x y z) \right)^2\\
	=&
	\Pi^{2}_\sigma
	\Big(
	\sum_{\alpha, \beta, \gamma, x,y,z} 
	A^{(1)}_{\alpha|x} B^{(1)}_{\beta|y} C^{(1)}_{\gamma|z} A^{(2)}_{\alpha|x} B^{(2)}_{\beta|y} C^{(2)}_{\gamma|z}\\ 
	&-2 p(\alpha\beta\gamma|xyz) A^{(1)}_{\alpha|x} B^{(1)}_{\beta|y} C^{(1)}_{\gamma|z} + p(\alpha\beta\gamma|xyz)^2 \Id
	\Big) \\
	=:&\Pi^{2}_\sigma(y_0),
\end{align*}
where $y_0$ is the element of $\D^2$ on which $\Pi_\sigma^2$ is evaluated.

Similarly, one can turn the independence constraint of \eqref{eq:bilocal_opt} into a linear constraint on two inflation levels.
However, it will turn out that for the completeness proof, it is necessary to impose constraints that are bounded from below and attain their minimal value on the feasible set of states.
We will thus formulate the factorization constraints as
\begin{align*}
	(\sigma(ac) -\sigma(a)\sigma(c))^2 = 0,
\end{align*}
so that the polarization becomes
\begin{align*}
	\Pi^{4}_\sigma (y_{ac}) = 0,
\end{align*}
where
\begin{multline*}
	y_{ac} := a^{(1)} c^{(1)} a^{(2)} c^{(2)} -2 a^{(1)} c^{(1)} a^{(2)} c^{(3)}\\ 
	+ a^{(1)} c^{(2)} a^{(3)} c^{(4)}
\end{multline*}
Here, the indices indicate which copies of the POVM elements are used to generate the operator, e.g.~$a^{(2)}$ can be written as (the norm limit of) a polynomial in the generators $\{\Id, A^{(2)}_{\alpha|x}\}$.
In this way, both the polynomial objective function and the polynomial constraints correspond to the linear pairing between operators $y_0, y_{ac}\in\D^4$ and symmetric product states in $K(\D^4)$.

More generally, given a degree $m$ polynomial $q$ whose action on states is bounded from below by 0, one can optimize over polynomial constraints of the form
\begin{align*}
	q(\sigma) = 0,
\end{align*}
by passing to the polarization $y_q \in \D^m$ of $q$.

Unfortunately, the set of symmetric product states is not an affine subset of state space, which means that the NPO method cannot directly optimize over it.
To get around this restriction, we will combine three tricks.
First, realize that NPO can optimize over the set of \emph{all} symmetric states.
Indeed, the symmetric group $S_n$ acts on $\D^n$ by permuting the indices of the generators, and a state $\rho\in K(\D^n)$ is \emph{symmetric} if it satisfies the \emph{linear} constraints $\rho(\pi(x)) = \rho(x)$ for every $x\in\D^n, \pi\in S_n$.
Second, in both quantum and classical probability \cite{diaconis1980finite, diaconis1987dozen, raggio1989quantum, caves2002unknown, navascues2020inflation, ligthart2021convergent}, there is a well-known family of statements collectively known as \emph{de~Finetti theorems} that show that symmetric states on infinitely many copies are a convex combination of symmetric product states.
In our particular case, ``infinitely many copies'' can be made rigorous as the \emph{inductive limit of maximal $C^*$-tensor products}.
The following de~Finetti theorem, adapted to this setting, is proven in Ref.~\cite{ligthart2021convergent}

\begin{theorem}[Max tensor product Quantum de Finetti Theorem \cite{ligthart2021convergent} ]\label{thm:finetti maximale}
	Let $\rho\in K(\D^\infty)$ be a symmetric state on an infinite maximal tensor product
	\begin{align*}
		\D^\infty = \lim_{n\to\infty} \D^{\otimes_{\max} n}.
	\end{align*}
	Then there exists a unique probability measure $\mu$ over states on $\D$ such that 
	for all $x\in\D^\infty$,
	\begin{align}\label{eqn:de finetti rep}
		\rho(x) = \int_{K(\D)} \Pi^{\infty}_\sigma(x) \,\mathrm{d}\mu(\sigma),
				\end{align}
			    where $\Pi_\sigma^{\infty}$ is the infinite symmetric product state on $\D^\infty$ associated with the state $\sigma$ on $\D$.
\end{theorem}

The third trick is to choose the polynomial constraints in such a way that they demand that point-wise non-negative polynomials are set to $0$.
If such an extremal condition is satisfied by a (continuous, as in Eq.~(\ref{eqn:de finetti rep})) convex combination, then it must in fact be satisfied almost surely.
Applying this to the constraints and the objective function, we will see that in our case the $\Pi^\infty_\sigma$ are almost surely a feasible solution of \eqref{eq:bilocal_opt} that attains the minimum $f^\infty = \lim_{n \to \infty} f^n$ of the relaxation \eqref{eq:bilocal_inflation_hierarchy} below.

Let us now formulate the NPO hierarchy and its convergence proof more precisely.
Define $\A^n$ to be the subalgebra of $\D^n$ that consists of Alice's operators and similar for Bob and Charlie.
Let $U$ be a countable basis of $\A^n$.
Usually this basis is taken to be the set of all words in Alice's POVM elements.
Define $V$ for Bob and $W$ for Charlie in a similar way.
For $n \geq 4$ the hierarchy of NPO problems is then given by
\begin{align}
	\begin{split} \label{eq:bilocal_inflation_hierarchy}
		f^n = \min_{\rho \in K({\D^n})} \ &\rho(y_0)\\
		\text{s. t. } \quad 
		& \rho(\pi(abc)) = \rho(abc), \\ 
		& \rho(y_{ac}) = 0, \\
		& \forall \pi\in S_n, a\in U, b\in V, c\in W.
	\end{split}
\end{align}
Each of these NPO problems can in turn be solved via the complete hierarchy of SDP relaxations introduced in Ref.~\cite{navascues2008convergent,pironio2010convergent}, where we have used the formulation of NPO problems in Ref.~\cite{ligthart2021convergent}.

The following theorem then states that \eqref{eq:bilocal_opt} and \eqref{eq:bilocal_inflation_hierarchy} are equivalent in the limit.
\begin{theorem} \label{thm:convergence}
Let $f^\infty = \lim_{n\to \infty} f^n$. It holds that $f^\infty = f^*$.
\end{theorem}
\begin{proof}
The proof is very similar to that of Theorem~11 in Ref.~\cite{ligthart2021convergent}.

It is clear that
\begin{align} \label{eq:convergence ->}
	f^n \leq f^* \qquad \forall n,
\end{align}
since each level of the hierarchy \eqref{eq:bilocal_inflation_hierarchy} is a relaxation of the optimization problem \eqref{eq:bilocal_opt}.

For the converse direction, use NPO to construct a state $\omega_n$ on $\D^\infty$ for each level $n$ of the hierarchy by taking the infinite tensor product of an optimizing state of the optimization problem \eqref{eq:bilocal_inflation_hierarchy} at level $n$.

By the Banach-Alaoglu theorem applied to the state space $K(\D^\infty)$, this sequence admits a weak$^*$-convergent subsequence.
Let $\omega$ be its limit point.
Since each $\omega_n$ obeys the constraints of Eq.~\eqref{eq:bilocal_inflation_hierarchy}, so does $\omega$.
Hence, $\omega$ is a symmetric state on the algebra $\D^\infty$ and Theorem \ref{thm:finetti maximale} applies.
That is, $\omega$ can be written as
\begin{align}\label{eq:omega_deFinetti}
	\omega = \int \dd \mu(\sigma)\ \Pi^\infty_\sigma,
\end{align}
with $\mu$ a unique probability measure over states $\sigma \in K(\D)$ and $\Pi^\infty_\sigma$ an infinite product state on $\D^\infty$.

By construction, each of the $y_{ac}$ is non-negative on the product states $\Pi^\infty_\sigma$.
Therefore, since $\omega(y_{ac}) = 0$, and $\mu$ is a probability measure, it holds that 
\begin{align*}
	\Pi^\infty_\sigma(y_{ac}) = 0 \qquad \text{almost everywhere w.r.t. } \mu.
\end{align*}
That is, there exists a full measure subset $E \subset K(\D)$ such that for all $\sigma \in E$, it holds that $\Pi^\infty_\sigma(y_{ac}) = 0$.

Hence, each $\Pi^\infty_\sigma$ with $\sigma \in E$ defines a feasible state $\sigma$ for the optimization problem \eqref{eq:bilocal_opt} by restricting to one copy of the algebra $\D$. 
From this one can conclude that $\Pi^\infty_\sigma (y_0) \geq f^\infty$ for all $\sigma \in E$, for otherwise one could have taken $\omega$ to be the point measure on a state $\sigma'$ such that $\Pi_{\sigma'}(y_0) < f^\infty$.
This would contradict the fact that $f^\infty$ is a minimum.

Combining this with the fact that $\omega(y_0) = f^\infty$, it must hold that
\begin{align*}
	\Pi^\infty_\sigma(y_0) = f^\infty \qquad \text{almost everywhere w.r.t. } \mu \text{ on } E.
\end{align*}
I.e., there exists a set $F \subset E$ with full measure, such that for all $\sigma \in F$ it holds that $\Pi^\infty_\sigma(y_0) = f^\infty$.
Finally, we can conclude for any $\sigma \in F$
\begin{align} \label{eq:convergence <-}
	f^\infty = \Pi^\infty_\sigma(y_0) \geq f^*.
\end{align}
Combining Eqs.~\eqref{eq:convergence ->} and \eqref{eq:convergence <-} yields $f^\infty = f^*$, proving the theorem.
\end{proof}

\subsection{Inflation hierarchy} \label{sec:inflation hierarchy}

\begin{figure}
\centering
\includegraphics[width=0.7\linewidth]{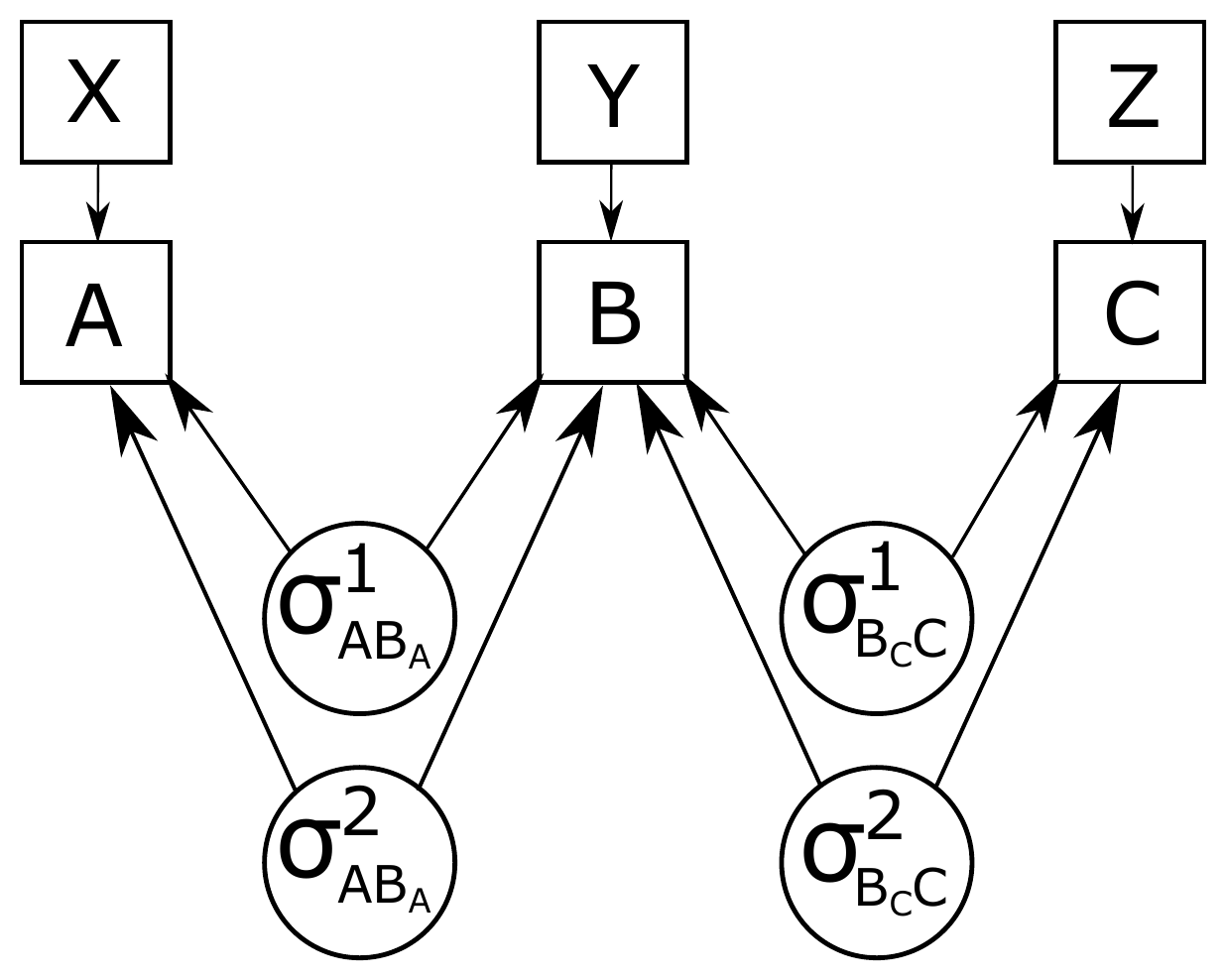} 
\caption{The level 2 inflation of the bilocal scenario. Each of the states $\sigma_{AB_A}$ and $\sigma_{B_C C}$ has been copied. The total state of the system is permutation symmetric under the exchange of each of these copies. The inflation technique builds on this observation.
}
\label{fig:bilocal_inflated}
\end{figure}

There exists a second convergent hierarchy that is more closely related to the quantum inflation hierarchy of Ref.~\cite{wolfe2021quantum}.
By showing convergence of such a hierarchy, we answer a question posed by Renou and Xu in Ref.~\cite{renou2022two}.
The hierarchy is very similar to that of Eq.~\eqref{eq:bilocal_inflation_hierarchy}, but instead of treating the independence constraints as polynomial conditions, they are enforced by imposing additional symmetries.
Loosely speaking,
these new symmetry constraints posit that copies of the state $\sigma_{AB_A}$ can be permuted independently of the copies of $\sigma_{B_C C}$,
see Fig.~\ref{fig:bilocal_inflated} for a visualization.
The advantage of this hierarchy over the polarization hierarchy is that the symmetry constraints can already be imposed at level 2 of the hierarchy.

In the notation introduced above Eq.~\eqref{eq:bilocal_inflation_hierarchy}, the level $n$ relaxation is given by
\begin{align}
\begin{split} \label{eq:hierarchy_symmetry}
	\tilde{f}^n = \min_{\rho \in K({\D^n})} \ &\rho(y_0)\\
		\text{s. t. } \quad 
		& \rho(\pi(abc)) = \rho(abc), \quad \forall \pi \in S_n\\
		& \rho(a \pi(c)) = \rho(ac),  \quad \forall \pi \in S_n.
\end{split}
\end{align}

\begin{theorem} \label{thm:convergence_inflation}
Let $\tilde{f}^\infty = \lim_{n\to\infty} \tilde{f}^n$. 
It holds that $\tilde{f}^\infty = f^*$.
\end{theorem}
\begin{proof}
We only give a short proof sketch, since the techniques are nearly identical to the proof of Theorem \ref{thm:convergence}.

Construct a state $\omega \in K(\D^\infty)$ as the limit of optimizing states of \eqref{eq:hierarchy_symmetry} (c.f. the proof of Theorem \ref{thm:convergence}).
By the de Finetti theorem this state has the form
\begin{align} \label{eq:omega_definetti_inflation}
\omega = \int \dd \mu (\sigma) \Pi^\infty_\sigma.
\end{align}

Fix one $n\in\NN$.
Using the cycle notation, define the permutation
\begin{align*}
	\pi &= (1,n+1) \, (2, n+2)\, \dots\, (n,2n),
\end{align*}
i.e.\ $\pi$ exchanges the 1st block of $n$ symbols with the 2nd block of $n$ symbols.
Using the additional symmetry constraints of $\omega$, when restricted to elements of Alice and Charlie, we see that for each $a \in \A^n$ and $c \in \C^n$, and for all $n$
\begin{align}
	\omega(ac) &= \omega(a \pi(c))  \label{eq:omega_sep_1}\\
	&= \int \dd \mu (\sigma)\ \Pi^\infty_\sigma (a \pi(c))  \label{eq:omega_sep_2}\\
	&= \int \dd \mu (\sigma)\ \Pi^\infty_\sigma(a) \Pi^\infty_\sigma(\pi(c))  \label{eq:omega_sep_3}\\
	&= \int \dd \mu (\sigma)\ \Pi^\infty_\sigma(a) \Pi^\infty_\sigma(c) \label{eq:omega_sep_4},
\end{align}
where the symmetry of $\omega$ was used in \eqref{eq:omega_sep_1}, Eq.~\eqref{eq:omega_definetti_inflation} was used for Eq.~\eqref{eq:omega_sep_2}, and in Eqs.~\eqref{eq:omega_sep_3} and \eqref{eq:omega_sep_4} it was used that each $\Pi^\infty_\sigma$ is a symmetric product state over disjoint inflation levels.
From this we can see that $\Pi^\infty_\sigma$ obeys the factorization constraint almost surely with respect to $\mu$.

The rest of the proof is now similar to that of Theorem \ref{thm:convergence}.
\end{proof}

We note that this result also proves convergence of the ``full'' quantum inflation hierarchy \cite{wolfe2021quantum} where Bob's POVMs have two separate indices: 
For each $n$, the NPO problem that describes such a full inflation level is a relaxation of problem \eqref{eq:bilocal_opt} that is at least as restrictive as the relaxation \eqref{eq:bilocal_inflation_hierarchy}. 
Hence, its optimal value lies between $f^n$ and $f^*$ for every $n$.

\section{Conclusion and discussion} \label{sec:conclusion}

In this paper we have shown the equivalence of several models of locality for the bilocal scenario.
In particular, we have shown that a reduced model of bilocality, in which only Alice and Charlie are supposed to be independent, is enough to reproduce exactly the bilocal quantum distributions in the commuting observables model.
Furthermore, if Alice's and Charlie's systems can be associated with a finite dimensional algebra, the correlations also coincide with the tensor product model.

Additionally, we have constructed two converging SDP hierarchies for the bilocal scenario, based on the above-mentioned classification.
The polarization hierarchy makes use of the fact that certain polynomial expressions in a state can be linearized on tensor powers of that state.
Here, this idea was applied to the factorization constraint between Alice and Charlie, but it can be applied to polynomials of higher order as well.
The second hierarchy is a form of the quantum inflation hierarchy.

In deriving these results, we have answered two open questions of Ref.~\cite{renou2022two}:
\begin{enumerate}
\item
whether the bilocal scenario allows for new insights into Tsirelson's problem: No.
\item
and whether the quantum inflation hierarchy is complete for the bilocal scenario: Yes.
\end{enumerate} 

Several follow-up questions suggest themselves.

One can ask whether it is possible to use the technique of Theorem \ref{thm:fact to biloc} to show that the quantum inflation hierarchy converges for other networks.
We believe Theorem \ref{thm:fact to biloc} can be adapted to the more general case of \emph{star networks}, in which one central party shares a bipartite quantum state with $n$ other parties, but no other connections are present.
Note that the bilocal scenario is a star network with $n=2$, where Bob acts as the central party.

The bilocal scenario is also a \emph{line network}, in which the parties are arranged in a line and share a bipartite quantum state with each of their neighbours.
It is less clear whether the technique can be extended to arbitrary line networks.

A numerical comparison can be made between the hierarchies suggested in this paper, the scalar extension hierarchy, and the quantum inflation hierarchy as originally suggested in Ref.~\cite{wolfe2021quantum}.
We leave such a comparison for later work.

\section{Acknowledgments}
We thank 
Johan \AA berg,
Mariami Gachechiladze,
Alejandro Pozas-Kerstjens,
Marc-Olivier Renou,
Xiangling Xu,
and
Julius Zeiss
for insightful discussions. 
This work has been supported by Germany's Excellence Strategy -- Cluster of Excellence Matter and Light for Quantum Computing (ML4Q) EXC 2004/1 -- 390534769.

{\bf Data sharing.} Data sharing is not applicable to this article as no new data were created or analyzed in this study.


\begin{thebibliography}{39}%
\makeatletter
\providecommand \@ifxundefined [1]{%
 \@ifx{#1\undefined}
}%
\providecommand \@ifnum [1]{%
 \ifnum #1\expandafter \@firstoftwo
 \else \expandafter \@secondoftwo
 \fi
}%
\providecommand \@ifx [1]{%
 \ifx #1\expandafter \@firstoftwo
 \else \expandafter \@secondoftwo
 \fi
}%
\providecommand \natexlab [1]{#1}%
\providecommand \enquote  [1]{``#1''}%
\providecommand \bibnamefont  [1]{#1}%
\providecommand \bibfnamefont [1]{#1}%
\providecommand \citenamefont [1]{#1}%
\providecommand \href@noop [0]{\@secondoftwo}%
\providecommand \href [0]{\begingroup \@sanitize@url \@href}%
\providecommand \@href[1]{\@@startlink{#1}\@@href}%
\providecommand \@@href[1]{\endgroup#1\@@endlink}%
\providecommand \@sanitize@url [0]{\catcode `\\12\catcode `\$12\catcode
  `\&12\catcode `\#12\catcode `\^12\catcode `\_12\catcode `\%12\relax}%
\providecommand \@@startlink[1]{}%
\providecommand \@@endlink[0]{}%
\providecommand \url  [0]{\begingroup\@sanitize@url \@url }%
\providecommand \@url [1]{\endgroup\@href {#1}{\urlprefix }}%
\providecommand \urlprefix  [0]{URL }%
\providecommand \Eprint [0]{\href }%
\providecommand \doibase [0]{https://doi.org/}%
\providecommand \selectlanguage [0]{\@gobble}%
\providecommand \bibinfo  [0]{\@secondoftwo}%
\providecommand \bibfield  [0]{\@secondoftwo}%
\providecommand \translation [1]{[#1]}%
\providecommand \BibitemOpen [0]{}%
\providecommand \bibitemStop [0]{}%
\providecommand \bibitemNoStop [0]{.\EOS\space}%
\providecommand \EOS [0]{\spacefactor3000\relax}%
\providecommand \BibitemShut  [1]{\csname bibitem#1\endcsname}%
\let\auto@bib@innerbib\@empty
\bibitem [{\citenamefont {Bell}(1964)}]{bell1964einstein}%
  \BibitemOpen
  \bibfield  {author} {\bibinfo {author} {\bibfnamefont {J.~S.}\ \bibnamefont
  {Bell}},\ }\bibfield  {title} {\bibinfo {title} {On the einstein podolsky
  rosen paradox},\ }\href@noop {} {\bibfield  {journal} {\bibinfo  {journal}
  {Physics Physique Fizika}\ }\textbf {\bibinfo {volume} {1}},\ \bibinfo
  {pages} {195} (\bibinfo {year} {1964})}\BibitemShut {NoStop}%
\bibitem [{\citenamefont {Bell}(2004)}]{bell2004speakable}%
  \BibitemOpen
  \bibfield  {author} {\bibinfo {author} {\bibfnamefont {J.~S.}\ \bibnamefont
  {Bell}},\ }\href@noop {} {\emph {\bibinfo {title} {Speakable and unspeakable
  in quantum mechanics: Collected papers on quantum philosophy}}}\ (\bibinfo
  {publisher} {Cambridge university press},\ \bibinfo {year}
  {2004})\BibitemShut {NoStop}%
\bibitem [{\citenamefont {Brunner}\ \emph {et~al.}(2014)\citenamefont
  {Brunner}, \citenamefont {Cavalcanti}, \citenamefont {Pironio}, \citenamefont
  {Scarani},\ and\ \citenamefont {Wehner}}]{brunner2014bell}%
  \BibitemOpen
  \bibfield  {author} {\bibinfo {author} {\bibfnamefont {N.}~\bibnamefont
  {Brunner}}, \bibinfo {author} {\bibfnamefont {D.}~\bibnamefont {Cavalcanti}},
  \bibinfo {author} {\bibfnamefont {S.}~\bibnamefont {Pironio}}, \bibinfo
  {author} {\bibfnamefont {V.}~\bibnamefont {Scarani}},\ and\ \bibinfo {author}
  {\bibfnamefont {S.}~\bibnamefont {Wehner}},\ }\bibfield  {title} {\bibinfo
  {title} {Bell nonlocality},\ }\href
  {https://doi.org/10.1103/RevModPhys.86.419} {\bibfield  {journal} {\bibinfo
  {journal} {Reviews of Modern Physics}\ }\textbf {\bibinfo {volume} {86}},\
  \bibinfo {pages} {419} (\bibinfo {year} {2014})}\BibitemShut {NoStop}%
\bibitem [{\citenamefont {Pan}\ \emph {et~al.}(2012)\citenamefont {Pan},
  \citenamefont {Chen}, \citenamefont {Lu}, \citenamefont {Weinfurter},
  \citenamefont {Zeilinger},\ and\ \citenamefont {\ifmmode~\dot{Z}\else
  \.{Z}\fi{}ukowski}}]{pan2012multiphoton}%
  \BibitemOpen
  \bibfield  {author} {\bibinfo {author} {\bibfnamefont {J.-W.}\ \bibnamefont
  {Pan}}, \bibinfo {author} {\bibfnamefont {Z.-B.}\ \bibnamefont {Chen}},
  \bibinfo {author} {\bibfnamefont {C.-Y.}\ \bibnamefont {Lu}}, \bibinfo
  {author} {\bibfnamefont {H.}~\bibnamefont {Weinfurter}}, \bibinfo {author}
  {\bibfnamefont {A.}~\bibnamefont {Zeilinger}},\ and\ \bibinfo {author}
  {\bibfnamefont {M.}~\bibnamefont {\ifmmode~\dot{Z}\else \.{Z}\fi{}ukowski}},\
  }\bibfield  {title} {\bibinfo {title} {Multiphoton entanglement and
  interferometry},\ }\href {https://doi.org/10.1103/RevModPhys.84.777}
  {\bibfield  {journal} {\bibinfo  {journal} {Rev. Mod. Phys.}\ }\textbf
  {\bibinfo {volume} {84}},\ \bibinfo {pages} {777} (\bibinfo {year}
  {2012})}\BibitemShut {NoStop}%
\bibitem [{\citenamefont {Ac{\'\i}n}\ \emph {et~al.}(2007)\citenamefont
  {Ac{\'\i}n}, \citenamefont {Brunner}, \citenamefont {Gisin}, \citenamefont
  {Massar}, \citenamefont {Pironio},\ and\ \citenamefont
  {Scarani}}]{acin2007device}%
  \BibitemOpen
  \bibfield  {author} {\bibinfo {author} {\bibfnamefont {A.}~\bibnamefont
  {Ac{\'\i}n}}, \bibinfo {author} {\bibfnamefont {N.}~\bibnamefont {Brunner}},
  \bibinfo {author} {\bibfnamefont {N.}~\bibnamefont {Gisin}}, \bibinfo
  {author} {\bibfnamefont {S.}~\bibnamefont {Massar}}, \bibinfo {author}
  {\bibfnamefont {S.}~\bibnamefont {Pironio}},\ and\ \bibinfo {author}
  {\bibfnamefont {V.}~\bibnamefont {Scarani}},\ }\bibfield  {title} {\bibinfo
  {title} {Device-independent security of quantum cryptography against
  collective attacks},\ }\href@noop {} {\bibfield  {journal} {\bibinfo
  {journal} {Physical Review Letters}\ }\textbf {\bibinfo {volume} {98}},\
  \bibinfo {pages} {230501} (\bibinfo {year} {2007})}\BibitemShut {NoStop}%
\bibitem [{\citenamefont {Barrett}\ \emph {et~al.}(2005)\citenamefont
  {Barrett}, \citenamefont {Hardy},\ and\ \citenamefont
  {Kent}}]{barrett2005no}%
  \BibitemOpen
  \bibfield  {author} {\bibinfo {author} {\bibfnamefont {J.}~\bibnamefont
  {Barrett}}, \bibinfo {author} {\bibfnamefont {L.}~\bibnamefont {Hardy}},\
  and\ \bibinfo {author} {\bibfnamefont {A.}~\bibnamefont {Kent}},\ }\bibfield
  {title} {\bibinfo {title} {No signaling and quantum key distribution},\
  }\href@noop {} {\bibfield  {journal} {\bibinfo  {journal} {Physical review
  letters}\ }\textbf {\bibinfo {volume} {95}},\ \bibinfo {pages} {010503}
  (\bibinfo {year} {2005})}\BibitemShut {NoStop}%
\bibitem [{\citenamefont {Pironio}\ \emph
  {et~al.}(2010{\natexlab{a}})\citenamefont {Pironio}, \citenamefont
  {Ac{\'\i}n}, \citenamefont {Massar}, \citenamefont {de~La~Giroday},
  \citenamefont {Matsukevich}, \citenamefont {Maunz}, \citenamefont
  {Olmschenk}, \citenamefont {Hayes}, \citenamefont {Luo}, \citenamefont
  {Manning} \emph {et~al.}}]{pironio2010random}%
  \BibitemOpen
  \bibfield  {author} {\bibinfo {author} {\bibfnamefont {S.}~\bibnamefont
  {Pironio}}, \bibinfo {author} {\bibfnamefont {A.}~\bibnamefont {Ac{\'\i}n}},
  \bibinfo {author} {\bibfnamefont {S.}~\bibnamefont {Massar}}, \bibinfo
  {author} {\bibfnamefont {A.~B.}\ \bibnamefont {de~La~Giroday}}, \bibinfo
  {author} {\bibfnamefont {D.~N.}\ \bibnamefont {Matsukevich}}, \bibinfo
  {author} {\bibfnamefont {P.}~\bibnamefont {Maunz}}, \bibinfo {author}
  {\bibfnamefont {S.}~\bibnamefont {Olmschenk}}, \bibinfo {author}
  {\bibfnamefont {D.}~\bibnamefont {Hayes}}, \bibinfo {author} {\bibfnamefont
  {L.}~\bibnamefont {Luo}}, \bibinfo {author} {\bibfnamefont {T.~A.}\
  \bibnamefont {Manning}}, \emph {et~al.},\ }\bibfield  {title} {\bibinfo
  {title} {Random numbers certified by bell’s theorem},\ }\href@noop {}
  {\bibfield  {journal} {\bibinfo  {journal} {Nature}\ }\textbf {\bibinfo
  {volume} {464}},\ \bibinfo {pages} {1021} (\bibinfo {year}
  {2010}{\natexlab{a}})}\BibitemShut {NoStop}%
\bibitem [{\citenamefont {Bancal}\ \emph {et~al.}(2011)\citenamefont {Bancal},
  \citenamefont {Gisin}, \citenamefont {Liang},\ and\ \citenamefont
  {Pironio}}]{bancal2011device}%
  \BibitemOpen
  \bibfield  {author} {\bibinfo {author} {\bibfnamefont {J.-D.}\ \bibnamefont
  {Bancal}}, \bibinfo {author} {\bibfnamefont {N.}~\bibnamefont {Gisin}},
  \bibinfo {author} {\bibfnamefont {Y.-C.}\ \bibnamefont {Liang}},\ and\
  \bibinfo {author} {\bibfnamefont {S.}~\bibnamefont {Pironio}},\ }\bibfield
  {title} {\bibinfo {title} {Device-independent witnesses of genuine
  multipartite entanglement},\ }\href
  {https://doi.org/10.1103/PhysRevLett.106.250404} {\bibfield  {journal}
  {\bibinfo  {journal} {Phys. Rev. Lett.}\ }\textbf {\bibinfo {volume} {106}},\
  \bibinfo {pages} {250404} (\bibinfo {year} {2011})}\BibitemShut {NoStop}%
\bibitem [{\citenamefont {Branciard}\ \emph {et~al.}(2012)\citenamefont
  {Branciard}, \citenamefont {Rosset}, \citenamefont {Gisin},\ and\
  \citenamefont {Pironio}}]{branciard2012bilocal}%
  \BibitemOpen
  \bibfield  {author} {\bibinfo {author} {\bibfnamefont {C.}~\bibnamefont
  {Branciard}}, \bibinfo {author} {\bibfnamefont {D.}~\bibnamefont {Rosset}},
  \bibinfo {author} {\bibfnamefont {N.}~\bibnamefont {Gisin}},\ and\ \bibinfo
  {author} {\bibfnamefont {S.}~\bibnamefont {Pironio}},\ }\bibfield  {title}
  {\bibinfo {title} {Bilocal versus nonbilocal correlations in
  entanglement-swapping experiments},\ }\href@noop {} {\bibfield  {journal}
  {\bibinfo  {journal} {Physical Review A}\ }\textbf {\bibinfo {volume} {85}},\
  \bibinfo {pages} {032119} (\bibinfo {year} {2012})}\BibitemShut {NoStop}%
\bibitem [{\citenamefont {Rosset}\ \emph {et~al.}(2016)\citenamefont {Rosset},
  \citenamefont {Branciard}, \citenamefont {Barnea}, \citenamefont {P{\"u}tz},
  \citenamefont {Brunner},\ and\ \citenamefont {Gisin}}]{rosset2016nonlinear}%
  \BibitemOpen
  \bibfield  {author} {\bibinfo {author} {\bibfnamefont {D.}~\bibnamefont
  {Rosset}}, \bibinfo {author} {\bibfnamefont {C.}~\bibnamefont {Branciard}},
  \bibinfo {author} {\bibfnamefont {T.~J.}\ \bibnamefont {Barnea}}, \bibinfo
  {author} {\bibfnamefont {G.}~\bibnamefont {P{\"u}tz}}, \bibinfo {author}
  {\bibfnamefont {N.}~\bibnamefont {Brunner}},\ and\ \bibinfo {author}
  {\bibfnamefont {N.}~\bibnamefont {Gisin}},\ }\bibfield  {title} {\bibinfo
  {title} {Nonlinear bell inequalities tailored for quantum networks},\
  }\href@noop {} {\bibfield  {journal} {\bibinfo  {journal} {Physical review
  letters}\ }\textbf {\bibinfo {volume} {116}},\ \bibinfo {pages} {010403}
  (\bibinfo {year} {2016})}\BibitemShut {NoStop}%
\bibitem [{\citenamefont {Pozas-Kerstjens}\ \emph {et~al.}(2022)\citenamefont
  {Pozas-Kerstjens}, \citenamefont {Gisin},\ and\ \citenamefont
  {Tavakoli}}]{pozas2022full}%
  \BibitemOpen
  \bibfield  {author} {\bibinfo {author} {\bibfnamefont {A.}~\bibnamefont
  {Pozas-Kerstjens}}, \bibinfo {author} {\bibfnamefont {N.}~\bibnamefont
  {Gisin}},\ and\ \bibinfo {author} {\bibfnamefont {A.}~\bibnamefont
  {Tavakoli}},\ }\bibfield  {title} {\bibinfo {title} {Full network
  nonlocality},\ }\href@noop {} {\bibfield  {journal} {\bibinfo  {journal}
  {Physical review letters}\ }\textbf {\bibinfo {volume} {128}},\ \bibinfo
  {pages} {010403} (\bibinfo {year} {2022})}\BibitemShut {NoStop}%
\bibitem [{\citenamefont {Tavakoli}\ \emph {et~al.}(2022)\citenamefont
  {Tavakoli}, \citenamefont {Pozas-Kerstjens}, \citenamefont {Luo},\ and\
  \citenamefont {Renou}}]{tavakoli2022bell}%
  \BibitemOpen
  \bibfield  {author} {\bibinfo {author} {\bibfnamefont {A.}~\bibnamefont
  {Tavakoli}}, \bibinfo {author} {\bibfnamefont {A.}~\bibnamefont
  {Pozas-Kerstjens}}, \bibinfo {author} {\bibfnamefont {M.-X.}\ \bibnamefont
  {Luo}},\ and\ \bibinfo {author} {\bibfnamefont {M.-O.}\ \bibnamefont
  {Renou}},\ }\bibfield  {title} {\bibinfo {title} {Bell nonlocality in
  networks},\ }\href@noop {} {\bibfield  {journal} {\bibinfo  {journal}
  {Reports on Progress in Physics}\ }\textbf {\bibinfo {volume} {85}},\
  \bibinfo {pages} {056001} (\bibinfo {year} {2022})}\BibitemShut {NoStop}%
\bibitem [{\citenamefont {\ifmmode~\dot{Z}\else \.{Z}\fi{}ukowski}\ \emph
  {et~al.}(1993)\citenamefont {\ifmmode~\dot{Z}\else \.{Z}\fi{}ukowski},
  \citenamefont {Zeilinger}, \citenamefont {Horne},\ and\ \citenamefont
  {Ekert}}]{zukowski1993swapping}%
  \BibitemOpen
  \bibfield  {author} {\bibinfo {author} {\bibfnamefont {M.}~\bibnamefont
  {\ifmmode~\dot{Z}\else \.{Z}\fi{}ukowski}}, \bibinfo {author} {\bibfnamefont
  {A.}~\bibnamefont {Zeilinger}}, \bibinfo {author} {\bibfnamefont {M.~A.}\
  \bibnamefont {Horne}},\ and\ \bibinfo {author} {\bibfnamefont {A.~K.}\
  \bibnamefont {Ekert}},\ }\bibfield  {title} {\bibinfo {title}
  {``event-ready-detectors'' bell experiment via entanglement swapping},\
  }\href {https://doi.org/10.1103/PhysRevLett.71.4287} {\bibfield  {journal}
  {\bibinfo  {journal} {Phys. Rev. Lett.}\ }\textbf {\bibinfo {volume} {71}},\
  \bibinfo {pages} {4287} (\bibinfo {year} {1993})}\BibitemShut {NoStop}%
\bibitem [{\citenamefont {Canabarro}\ \emph {et~al.}(2019)\citenamefont
  {Canabarro}, \citenamefont {Brito},\ and\ \citenamefont
  {Chaves}}]{canabarro2019machine}%
  \BibitemOpen
  \bibfield  {author} {\bibinfo {author} {\bibfnamefont {A.}~\bibnamefont
  {Canabarro}}, \bibinfo {author} {\bibfnamefont {S.}~\bibnamefont {Brito}},\
  and\ \bibinfo {author} {\bibfnamefont {R.}~\bibnamefont {Chaves}},\
  }\bibfield  {title} {\bibinfo {title} {Machine learning nonlocal
  correlations},\ }\href@noop {} {\bibfield  {journal} {\bibinfo  {journal}
  {Physical review letters}\ }\textbf {\bibinfo {volume} {122}},\ \bibinfo
  {pages} {200401} (\bibinfo {year} {2019})}\BibitemShut {NoStop}%
\bibitem [{\citenamefont {Chaves}\ \emph {et~al.}(2015)\citenamefont {Chaves},
  \citenamefont {Majenz},\ and\ \citenamefont {Gross}}]{chaves2015information}%
  \BibitemOpen
  \bibfield  {author} {\bibinfo {author} {\bibfnamefont {R.}~\bibnamefont
  {Chaves}}, \bibinfo {author} {\bibfnamefont {C.}~\bibnamefont {Majenz}},\
  and\ \bibinfo {author} {\bibfnamefont {D.}~\bibnamefont {Gross}},\ }\bibfield
   {title} {\bibinfo {title} {Information--theoretic implications of quantum
  causal structures},\ }\href {https://doi.org/10.1038/ncomms6766} {\bibfield
  {journal} {\bibinfo  {journal} {Nature communications}\ }\textbf {\bibinfo
  {volume} {6}},\ \bibinfo {pages} {1} (\bibinfo {year} {2015})}\BibitemShut
  {NoStop}%
\bibitem [{\citenamefont {Pozas~Kerstjens}(2019)}]{pozas2019quantum}%
  \BibitemOpen
  \bibfield  {author} {\bibinfo {author} {\bibfnamefont {A.}~\bibnamefont
  {Pozas~Kerstjens}},\ }\href@noop {} {\bibinfo {title} {Quantum information
  outside quantum information}} (\bibinfo {year} {2019})\BibitemShut {NoStop}%
\bibitem [{\citenamefont {Pozas-Kerstjens}\ \emph {et~al.}(2019)\citenamefont
  {Pozas-Kerstjens}, \citenamefont {Rabelo}, \citenamefont {Rudnicki},
  \citenamefont {Chaves}, \citenamefont {Cavalcanti}, \citenamefont
  {Navascu\'es},\ and\ \citenamefont {Ac\'{\i}n}}]{pozas2019bounding}%
  \BibitemOpen
  \bibfield  {author} {\bibinfo {author} {\bibfnamefont {A.}~\bibnamefont
  {Pozas-Kerstjens}}, \bibinfo {author} {\bibfnamefont {R.}~\bibnamefont
  {Rabelo}}, \bibinfo {author} {\bibfnamefont {L.}~\bibnamefont {Rudnicki}},
  \bibinfo {author} {\bibfnamefont {R.}~\bibnamefont {Chaves}}, \bibinfo
  {author} {\bibfnamefont {D.}~\bibnamefont {Cavalcanti}}, \bibinfo {author}
  {\bibfnamefont {M.}~\bibnamefont {Navascu\'es}},\ and\ \bibinfo {author}
  {\bibfnamefont {A.}~\bibnamefont {Ac\'{\i}n}},\ }\bibfield  {title} {\bibinfo
  {title} {Bounding the sets of classical and quantum correlations in
  networks},\ }\href {https://doi.org/10.1103/PhysRevLett.123.140503}
  {\bibfield  {journal} {\bibinfo  {journal} {Phys. Rev. Lett.}\ }\textbf
  {\bibinfo {volume} {123}},\ \bibinfo {pages} {140503} (\bibinfo {year}
  {2019})}\BibitemShut {NoStop}%
\bibitem [{\citenamefont {Renou}\ and\ \citenamefont
  {Xu}(2022)}]{renou2022two}%
  \BibitemOpen
  \bibfield  {author} {\bibinfo {author} {\bibfnamefont {M.-O.}\ \bibnamefont
  {Renou}}\ and\ \bibinfo {author} {\bibfnamefont {X.}~\bibnamefont {Xu}},\
  }\bibfield  {title} {\bibinfo {title} {Two convergent npa-like hierarchies
  for the quantum bilocal scenario},\ }\href@noop {} {\bibfield  {journal}
  {\bibinfo  {journal} {arXiv preprint arXiv:2210.09065v2}\ } (\bibinfo {year}
  {2022})}\BibitemShut {NoStop}%
\bibitem [{\citenamefont {Wolfe}\ \emph {et~al.}(2019)\citenamefont {Wolfe},
  \citenamefont {Spekkens},\ and\ \citenamefont {Fritz}}]{wolfe2019inflation}%
  \BibitemOpen
  \bibfield  {author} {\bibinfo {author} {\bibfnamefont {E.}~\bibnamefont
  {Wolfe}}, \bibinfo {author} {\bibfnamefont {R.~W.}\ \bibnamefont
  {Spekkens}},\ and\ \bibinfo {author} {\bibfnamefont {T.}~\bibnamefont
  {Fritz}},\ }\bibfield  {title} {\bibinfo {title} {The inflation technique for
  causal inference with latent variables},\ }\bibfield  {journal} {\bibinfo
  {journal} {Journal of Causal Inference}\ }\textbf {\bibinfo {volume} {7}},\
  \href {https://doi.org/10.1515/jci-2017-0020} {10.1515/jci-2017-0020}
  (\bibinfo {year} {2019})\BibitemShut {NoStop}%
\bibitem [{\citenamefont {Navascu{\'e}s}\ and\ \citenamefont
  {Wolfe}(2020)}]{navascues2020inflation}%
  \BibitemOpen
  \bibfield  {author} {\bibinfo {author} {\bibfnamefont {M.}~\bibnamefont
  {Navascu{\'e}s}}\ and\ \bibinfo {author} {\bibfnamefont {E.}~\bibnamefont
  {Wolfe}},\ }\bibfield  {title} {\bibinfo {title} {The inflation technique
  completely solves the causal compatibility problem},\ }\href
  {https://doi.org/10.1515/jci-2018-0008} {\bibfield  {journal} {\bibinfo
  {journal} {Journal of Causal Inference}\ }\textbf {\bibinfo {volume} {8}},\
  \bibinfo {pages} {70} (\bibinfo {year} {2020})}\BibitemShut {NoStop}%
\bibitem [{\citenamefont {Wolfe}\ \emph {et~al.}(2021)\citenamefont {Wolfe},
  \citenamefont {Pozas-Kerstjens}, \citenamefont {Grinberg}, \citenamefont
  {Rosset}, \citenamefont {Ac{\'\i}n},\ and\ \citenamefont
  {Navascu{\'e}s}}]{wolfe2021quantum}%
  \BibitemOpen
  \bibfield  {author} {\bibinfo {author} {\bibfnamefont {E.}~\bibnamefont
  {Wolfe}}, \bibinfo {author} {\bibfnamefont {A.}~\bibnamefont
  {Pozas-Kerstjens}}, \bibinfo {author} {\bibfnamefont {M.}~\bibnamefont
  {Grinberg}}, \bibinfo {author} {\bibfnamefont {D.}~\bibnamefont {Rosset}},
  \bibinfo {author} {\bibfnamefont {A.}~\bibnamefont {Ac{\'\i}n}},\ and\
  \bibinfo {author} {\bibfnamefont {M.}~\bibnamefont {Navascu{\'e}s}},\
  }\bibfield  {title} {\bibinfo {title} {Quantum inflation: A general approach
  to quantum causal compatibility},\ }\href
  {https://doi.org/10.1103/PhysRevX.11.021043} {\bibfield  {journal} {\bibinfo
  {journal} {Physical Review X}\ }\textbf {\bibinfo {volume} {11}},\ \bibinfo
  {pages} {021043} (\bibinfo {year} {2021})}\BibitemShut {NoStop}%
\bibitem [{\citenamefont {Ligthart}\ \emph {et~al.}(2021)\citenamefont
  {Ligthart}, \citenamefont {Gachechiladze},\ and\ \citenamefont
  {Gross}}]{ligthart2021convergent}%
  \BibitemOpen
  \bibfield  {author} {\bibinfo {author} {\bibfnamefont {L.~T.}\ \bibnamefont
  {Ligthart}}, \bibinfo {author} {\bibfnamefont {M.}~\bibnamefont
  {Gachechiladze}},\ and\ \bibinfo {author} {\bibfnamefont {D.}~\bibnamefont
  {Gross}},\ }\bibfield  {title} {\bibinfo {title} {A convergent inflation
  hierarchy for quantum causal structures},\ }\href@noop {} {\bibfield
  {journal} {\bibinfo  {journal} {arXiv preprint arXiv:2110.14659}\ } (\bibinfo
  {year} {2021})}\BibitemShut {NoStop}%
\bibitem [{\citenamefont {Navascu{\'e}s}\ \emph {et~al.}(2008)\citenamefont
  {Navascu{\'e}s}, \citenamefont {Pironio},\ and\ \citenamefont
  {Ac{\'\i}n}}]{navascues2008convergent}%
  \BibitemOpen
  \bibfield  {author} {\bibinfo {author} {\bibfnamefont {M.}~\bibnamefont
  {Navascu{\'e}s}}, \bibinfo {author} {\bibfnamefont {S.}~\bibnamefont
  {Pironio}},\ and\ \bibinfo {author} {\bibfnamefont {A.}~\bibnamefont
  {Ac{\'\i}n}},\ }\bibfield  {title} {\bibinfo {title} {A convergent hierarchy
  of semidefinite programs characterizing the set of quantum correlations},\
  }\href {https://doi.org/10.1088/1367-2630/10/7/073013} {\bibfield  {journal}
  {\bibinfo  {journal} {New Journal of Physics}\ }\textbf {\bibinfo {volume}
  {10}},\ \bibinfo {pages} {073013} (\bibinfo {year} {2008})}\BibitemShut
  {NoStop}%
\bibitem [{\citenamefont {Pironio}\ \emph
  {et~al.}(2010{\natexlab{b}})\citenamefont {Pironio}, \citenamefont
  {Navascu{\'e}s},\ and\ \citenamefont {Acin}}]{pironio2010convergent}%
  \BibitemOpen
  \bibfield  {author} {\bibinfo {author} {\bibfnamefont {S.}~\bibnamefont
  {Pironio}}, \bibinfo {author} {\bibfnamefont {M.}~\bibnamefont
  {Navascu{\'e}s}},\ and\ \bibinfo {author} {\bibfnamefont {A.}~\bibnamefont
  {Acin}},\ }\bibfield  {title} {\bibinfo {title} {Convergent relaxations of
  polynomial optimization problems with noncommuting variables},\ }\href
  {https://doi.org/10.1137/090760155} {\bibfield  {journal} {\bibinfo
  {journal} {SIAM Journal on Optimization}\ }\textbf {\bibinfo {volume} {20}},\
  \bibinfo {pages} {2157} (\bibinfo {year} {2010}{\natexlab{b}})}\BibitemShut
  {NoStop}%
\bibitem [{\citenamefont {Raggio}\ and\ \citenamefont
  {Werner}(1989)}]{raggio1989quantum}%
  \BibitemOpen
  \bibfield  {author} {\bibinfo {author} {\bibfnamefont {G.}~\bibnamefont
  {Raggio}}\ and\ \bibinfo {author} {\bibfnamefont {R.}~\bibnamefont
  {Werner}},\ }\bibfield  {title} {\bibinfo {title} {Quantum statistical
  mechanics of general mean field systems},\ }\href@noop {} {\bibfield
  {journal} {\bibinfo  {journal} {Helv. Phys. Acta}\ }\textbf {\bibinfo
  {volume} {62}} (\bibinfo {year} {1989})}\BibitemShut {NoStop}%
\bibitem [{\citenamefont {Bratteli}\ and\ \citenamefont
  {Robinson}(2012)}]{bratteli2012operator}%
  \BibitemOpen
  \bibfield  {author} {\bibinfo {author} {\bibfnamefont {O.}~\bibnamefont
  {Bratteli}}\ and\ \bibinfo {author} {\bibfnamefont {D.~W.}\ \bibnamefont
  {Robinson}},\ }\href@noop {} {\emph {\bibinfo {title} {Operator Algebras and
  Quantum Statistical Mechanics: Volume 1: {$C^*$}-and {$W^*$}-Algebras.
  Symmetry Groups. Decomposition of States}}}\ (\bibinfo  {publisher}
  {Springer},\ \bibinfo {address} {Berlin},\ \bibinfo {year}
  {2012})\BibitemShut {NoStop}%
\bibitem [{\citenamefont {Landsman}(2017)}]{landsman2017foundations}%
  \BibitemOpen
  \bibfield  {author} {\bibinfo {author} {\bibfnamefont {K.}~\bibnamefont
  {Landsman}},\ }\href@noop {} {\emph {\bibinfo {title} {Foundations of quantum
  theory: From classical concepts to operator algebras}}}\ (\bibinfo
  {publisher} {Springer Nature},\ \bibinfo {year} {2017})\BibitemShut {NoStop}%
\bibitem [{\citenamefont {Moretti}(2019)}]{moretti2019fundamental}%
  \BibitemOpen
  \bibfield  {author} {\bibinfo {author} {\bibfnamefont {V.}~\bibnamefont
  {Moretti}},\ }\href@noop {} {\emph {\bibinfo {title} {Fundamental
  Mathematical Structures of Quantum Theory: Spectral Theory, Foundational
  Issues, Symmetries, Algebraic Formulation}}}\ (\bibinfo  {publisher}
  {Springer},\ \bibinfo {year} {2019})\BibitemShut {NoStop}%
\bibitem [{\citenamefont {Blackadar}(2006)}]{blackadar2006operator}%
  \BibitemOpen
  \bibfield  {author} {\bibinfo {author} {\bibfnamefont {B.}~\bibnamefont
  {Blackadar}},\ }\href@noop {} {\emph {\bibinfo {title} {Operator algebras:
  theory of {$C^*$}-algebras and von Neumann algebras}}},\ Vol.\ \bibinfo
  {volume} {122}\ (\bibinfo  {publisher} {Springer Science \& Business Media},\
  \bibinfo {address} {Berlin},\ \bibinfo {year} {2006})\BibitemShut {NoStop}%
\bibitem [{\citenamefont {Takesaki}(2001)}]{takesaki1}%
  \BibitemOpen
  \bibfield  {author} {\bibinfo {author} {\bibfnamefont {M.}~\bibnamefont
  {Takesaki}},\ }\href@noop {} {\emph {\bibinfo {title} {Theory of operator
  algebras I}}},\ Vol.\ \bibinfo {volume} {124}\ (\bibinfo  {publisher}
  {Springer Science \& Business Media},\ \bibinfo {address} {Berlin},\ \bibinfo
  {year} {2001})\BibitemShut {NoStop}%
\bibitem [{\citenamefont {Scholz}\ and\ \citenamefont
  {Werner}(2008)}]{scholz2008tsirelson}%
  \BibitemOpen
  \bibfield  {author} {\bibinfo {author} {\bibfnamefont {V.~B.}\ \bibnamefont
  {Scholz}}\ and\ \bibinfo {author} {\bibfnamefont {R.~F.}\ \bibnamefont
  {Werner}},\ }\bibfield  {title} {\bibinfo {title} {{T}sirelson's problem},\
  }\href@noop {} {\bibfield  {journal} {\bibinfo  {journal} {arXiv:0812.4305}\
  } (\bibinfo {year} {2008})}\BibitemShut {NoStop}%
\bibitem [{\citenamefont {Junge}\ \emph {et~al.}(2011)\citenamefont {Junge},
  \citenamefont {Navascues}, \citenamefont {Palazuelos}, \citenamefont
  {Perez-Garcia}, \citenamefont {Scholz},\ and\ \citenamefont
  {Werner}}]{junge2011connes}%
  \BibitemOpen
  \bibfield  {author} {\bibinfo {author} {\bibfnamefont {M.}~\bibnamefont
  {Junge}}, \bibinfo {author} {\bibfnamefont {M.}~\bibnamefont {Navascues}},
  \bibinfo {author} {\bibfnamefont {C.}~\bibnamefont {Palazuelos}}, \bibinfo
  {author} {\bibfnamefont {D.}~\bibnamefont {Perez-Garcia}}, \bibinfo {author}
  {\bibfnamefont {V.~B.}\ \bibnamefont {Scholz}},\ and\ \bibinfo {author}
  {\bibfnamefont {R.~F.}\ \bibnamefont {Werner}},\ }\bibfield  {title}
  {\bibinfo {title} {Connes' embedding problem and {T}sirelson's problem},\
  }\href {https://doi.org/10.1063/1.3514538} {\bibfield  {journal} {\bibinfo
  {journal} {Journal of Mathematical Physics}\ }\textbf {\bibinfo {volume}
  {52}},\ \bibinfo {pages} {012102} (\bibinfo {year} {2011})}\BibitemShut
  {NoStop}%
\bibitem [{\citenamefont {Fritz}(2012)}]{fritz2012tsirelson}%
  \BibitemOpen
  \bibfield  {author} {\bibinfo {author} {\bibfnamefont {T.}~\bibnamefont
  {Fritz}},\ }\bibfield  {title} {\bibinfo {title} {Tsirelson's problem and
  {K}irchberg's conjecture},\ }\href
  {https://doi.org/10.1142/S0129055X12500122} {\bibfield  {journal} {\bibinfo
  {journal} {Reviews in Mathematical Physics}\ }\textbf {\bibinfo {volume}
  {24}},\ \bibinfo {pages} {1250012} (\bibinfo {year} {2012})}\BibitemShut
  {NoStop}%
\bibitem [{\citenamefont {Ji}\ \emph {et~al.}(2020)\citenamefont {Ji},
  \citenamefont {Natarajan}, \citenamefont {Vidick}, \citenamefont {Wright},\
  and\ \citenamefont {Yuen}}]{ji2020mip}%
  \BibitemOpen
  \bibfield  {author} {\bibinfo {author} {\bibfnamefont {Z.}~\bibnamefont
  {Ji}}, \bibinfo {author} {\bibfnamefont {A.}~\bibnamefont {Natarajan}},
  \bibinfo {author} {\bibfnamefont {T.}~\bibnamefont {Vidick}}, \bibinfo
  {author} {\bibfnamefont {J.}~\bibnamefont {Wright}},\ and\ \bibinfo {author}
  {\bibfnamefont {H.}~\bibnamefont {Yuen}},\ }\bibfield  {title} {\bibinfo
  {title} {{$\mathrm{MIP}^*= \mathrm{RE}$}},\ }\href@noop {} {\bibfield
  {journal} {\bibinfo  {journal} {arXiv:2001.04383}\ } (\bibinfo {year}
  {2020})}\BibitemShut {NoStop}%
\bibitem [{\citenamefont {Kadison}\ and\ \citenamefont
  {Ringrose}(1983)}]{kadison1983fundamentals}%
  \BibitemOpen
  \bibfield  {author} {\bibinfo {author} {\bibfnamefont {R.~V.}\ \bibnamefont
  {Kadison}}\ and\ \bibinfo {author} {\bibfnamefont {J.~R.}\ \bibnamefont
  {Ringrose}},\ }\href@noop {} {\emph {\bibinfo {title} {Fundamentals of the
  Theory of Operator Algebras. Vol. I}}}\ (\bibinfo  {publisher} {Oxford
  University Press},\ \bibinfo {year} {1983})\BibitemShut {NoStop}%
\bibitem [{\citenamefont {Tsirelson}(2007)}]{tsirelson2007open}%
  \BibitemOpen
  \bibfield  {author} {\bibinfo {author} {\bibfnamefont {B.}~\bibnamefont
  {Tsirelson}},\ }\href
  {https://web.archive.org/web/20090414083019/http://www.imaph.tu-bs.de/qi/problems/33.html}
  {\bibinfo {title} {Bell inequalities and operator algebras (open problem 33
  in quantum information)}} (\bibinfo {year} {2007})\BibitemShut {NoStop}%
\bibitem [{\citenamefont {Diaconis}\ and\ \citenamefont
  {Freedman}(1980)}]{diaconis1980finite}%
  \BibitemOpen
  \bibfield  {author} {\bibinfo {author} {\bibfnamefont {P.}~\bibnamefont
  {Diaconis}}\ and\ \bibinfo {author} {\bibfnamefont {D.}~\bibnamefont
  {Freedman}},\ }\bibfield  {title} {\bibinfo {title} {Finite exchangeable
  sequences},\ }\href@noop {} {\bibfield  {journal} {\bibinfo  {journal} {The
  Annals of Probability}\ ,\ \bibinfo {pages} {745}} (\bibinfo {year}
  {1980})}\BibitemShut {NoStop}%
\bibitem [{\citenamefont {Diaconis}\ and\ \citenamefont
  {Freedman}(1987)}]{diaconis1987dozen}%
  \BibitemOpen
  \bibfield  {author} {\bibinfo {author} {\bibfnamefont {P.}~\bibnamefont
  {Diaconis}}\ and\ \bibinfo {author} {\bibfnamefont {D.}~\bibnamefont
  {Freedman}},\ }\bibfield  {title} {\bibinfo {title} {A dozen de
  {F}inetti-style results in search of a theory},\ }in\ \href@noop {} {\emph
  {\bibinfo {booktitle} {Annales de l'IHP Probabilit{\'e}s et statistiques}}},\
  Vol.~\bibinfo {volume} {23}\ (\bibinfo {year} {1987})\ pp.\ \bibinfo {pages}
  {397--423}\BibitemShut {NoStop}%
\bibitem [{\citenamefont {Caves}\ \emph {et~al.}(2002)\citenamefont {Caves},
  \citenamefont {Fuchs},\ and\ \citenamefont {Schack}}]{caves2002unknown}%
  \BibitemOpen
  \bibfield  {author} {\bibinfo {author} {\bibfnamefont {C.~M.}\ \bibnamefont
  {Caves}}, \bibinfo {author} {\bibfnamefont {C.~A.}\ \bibnamefont {Fuchs}},\
  and\ \bibinfo {author} {\bibfnamefont {R.}~\bibnamefont {Schack}},\
  }\bibfield  {title} {\bibinfo {title} {Unknown quantum states: the quantum de
  finetti representation},\ }\href {https://doi.org/10.1063/1.1494475}
  {\bibfield  {journal} {\bibinfo  {journal} {Journal of Mathematical Physics}\
  }\textbf {\bibinfo {volume} {43}},\ \bibinfo {pages} {4537} (\bibinfo {year}
  {2002})}\BibitemShut {NoStop}%
\end{thebibliography}
\end{document}